\documentclass[11pt,a4paper]{article}

\usepackage{graphicx}%
\usepackage{multirow}%
\usepackage{amsmath,amssymb,amsfonts}%
\usepackage{amsthm}%
\usepackage{mathrsfs}%
\usepackage[title]{appendix}%
\usepackage{xcolor}%
\usepackage{textcomp}%
\usepackage{manyfoot}%
\usepackage{booktabs}%
\usepackage{algorithm}%
\usepackage{algorithmicx}%
\usepackage{algpseudocode}%
\usepackage{listings}%
\numberwithin{equation}{section}
\usepackage{lipsum}
\usepackage{bbm}
\usepackage{mathtools}
\usepackage{blindtext}
\usepackage[utf8]{inputenc}
\usepackage[T1]{fontenc}
\usepackage{mathrsfs}
\usepackage{nicefrac}
\usepackage{microtype}
\usepackage{doi}
\usepackage{bm}
\usepackage{enumitem}
\usepackage[a4paper,margin=1.4in]{geometry}
\usepackage{lmodern}
\usepackage[numbers]{natbib}

\newtheoremstyle{newRemark}%
    {10pt}%
    {10pt}%
    {}
    {}
    {\itshape}
    {:}
    {.5em}
    {}%

\newtheoremstyle{newPlain}%
    {10pt}%
    {10pt}%
    {\itshape}
    {}
    {\bfseries}
    {:}
    {.5em}
    {}%

\newtheoremstyle{newDefinition}%
    {10pt}%
    {10pt}%
    {\itshape}
    {}
    {\bfseries}
    {:}
    {.5em}
    {}%

\theoremstyle{newPlain}
\newtheorem{theorem}{Theorem}[section]
\newtheorem{lemma}[theorem]{Lemma}
\newtheorem{proposition}[theorem]{Proposition}

\theoremstyle{newDefinition}
\newtheorem{definition}[theorem]{Definition}
\newtheorem{assumption}{Assumption}[theorem]

\theoremstyle{newRemark}
\newtheorem{remark}[theorem]{Remark}
\newtheorem{example}[theorem]{Example}

\newcounter{notecount}[section] 
\renewcommand{\thenotecount}{\thesection.\arabic{notecount}}

\begin{document}

\title{European Options in Market Models with Multiple Defaults: the BSDE approach}

\author{
  Miryana Grigorova\thanks{\textit{Corresponding Author}, Email: \texttt{miryana.grigorova@warwick.ac.uk},  Department of Statistics, University of Warwick} \and
  James Wheeldon\thanks{Email: \texttt{james.wheeldon@warwick.ac.uk}, Department of Statistics, University of Warwick\\\textit{Acknowledgements:} We thank Marie-Claire Quenez for helpful discussions.} 
}
 
\date{\today}  

\maketitle

\begin{abstract}
We study non-linear Backward Stochastic Differential Equations (BSDEs) driven by a Brownian motion and $p$ default martingales. The driver of the BSDE with multiple default jumps can take a generalized form involving an optional finite variation process.  We first show existence and uniqueness. We then  establish comparison and strict comparison results for these BSDEs, under a suitable assumption on the driver. In the case of a linear driver, we derive an explicit formula for the first component of the BSDE using an adjoint exponential semimartingale. The representation depends on whether the finite variation process is predictable or only optional. We apply our results  to the problem of pricing and hedging a European option in a linear complete market with two defaultable assets and in  a non-linear complete market with $p$ defaultable assets. Two examples of the latter market model are provided: an example where the seller of the option is a large investor  influencing the probability of default of a single asset and  an example   where the large seller’s strategy affects the default probabilities of all $p$ assets.
\end{abstract}

\bigskip


\newpage

\section{Introduction}\label{intro}

In this paper, we consider BSDEs with multiple default jumps and explore some applications in a financial context. BSDEs were first introduced by Bismut \cite{bismut1973conjugate}, who studied the linear case in relation to stochastic control. Pardoux and Peng \cite{pardoux1990adapted} established the well-posedness of non-linear BSDEs with Lipschitz drivers in a Brownian filtration.

BSDEs incorporating both continuous and jump components have also been studied in the literature (cf., e.g., \cite{royer2006bsdejumps}, \cite{quenez2013bsdejumps}, \cite{becherer2017monotone},  \cite{dumitrescu2018bsdes}, \cite{papapantoleon2016existence}). 
In \cite{dumitrescu2018bsdes}, BSDEs with a single default jump are studied and applied in a financial context (for developments in the incomplete single default framework, we refer to \cite{grigorova2020european, grigorova2021american}). 

In the present paper, we consider BSDEs driven by a Brownian motion and $p$ compensated default martingales, each associated with a default time $\tau_i > 0$ and an intensity process $\lambda^i = (\lambda_t^i)_{t \in [0, T]}$ for $i \in \{1, \ldots, p\}$. We develop new \emph{a priori} estimates and establish the existence and uniqueness of the solution to the BSDE with multiple default jumps.

We consider an optional (not necessarily predictable), right-continuous left-limited (rcll) process $D$ of finite variation which enters into the driver of the BSDE, leading to a generalized form of the driver: $g(t, y, z, k^1, \ldots, k^p)dt + dD_t$. This modeling choice is motivated by the fact that, in markets with defaultable securities, contingent claims often give rise to intermediate cash flows, particularly at the time of default, as observed in \cite{bielecki2004hedging}.

We extend the definition of $\lambda$-linear drivers introduced in \cite{dumitrescu2018bsdes} to the general case of $\lambda^{(p)}$-linear drivers, where $\lambda^{(p)}$ refers to the vector of $p$ intensity processes $\lambda^1, \ldots, \lambda^p$. When the driver $g$ is $\lambda^{(p)}$-linear, we derive an explicit representation of the solution to the BSDE with multiple default jumps using an associated adjoint semimartingale. This representation depends on whether the process $D$ is predictable or only optional. We also prove a comparison result and a strict comparison result under suitable assumptions on the driver, where we distinguish again the case where  $D$ is predictable and the case where $D$ is  only optional.

We present two financial applications of this model: one in which the market is linear and complete, and another in which the market is complete but non-linear. In both cases, we assume the existence of a risk-free asset, a default-free (jump-free) risky asset, and two or more defaultable assets. We focus on pricing and hedging a European contingent claim with terminal payoff $\eta$ at maturity $T > 0$, and with intermediate `dividends' modeled by the process $D$, where $D$ represents an exogenous cumulative process. The process $D$ is not necessarily predictable.  In the second example, the market is non-linear, with the non-linearity arising from imperfections caused by a large seller whose strategy influences default probabilities.

In the non-linear market setting with $p$ defaults, we show that the price of the contingent claim is $X_{\cdot, T}^g(\eta,D)$, where $X_{\cdot, T}^g(\eta,D)$ denotes the solution to the non-linear BSDE with multiple default jumps, terminal time $T$, terminal condition $\eta$, and   generalized driver of the form $g(t,y,z,k^1,\ldots,k^p)dt + dD_t$. This gives rise to a non-linear pricing system $\mathbf{X}^g:(\eta,D)\mapsto X_{\cdot,T}^g(\eta,D)$, whose properties we study. When $D$ is fixed, we define the associated $(g,D)$-conditional evaluation by $\mathscr{E}_{\cdot,T}^{g,D}(\eta)\coloneqq X_{\cdot,T}^g(\eta,D)$ for  $\eta\in L^2(\mathcal{G}_T),$ and provide its main properties. 

The remainder of the paper is organized as follows: \\
Section \ref{sec:PPS} is dedicated to the study of the BSDE with $p$  default jumps and generalized driver. 
In Subsection \ref{sub-section_BSDE}, we define the BSDE with $p$  default jumps and generalized  $\lambda^{(p)}$-driver  (allowing for an intermediate optional finite variational process $D$);  in Subsection \ref{subsection_existence}, we prove existence and uniqueness, and study the case where the generalized  $\lambda^{(p)}$-driver is linear; in   Subsection \ref{subsection_comparison}, we establish comparison and  strict comparison results under  suitable assumptions on the driver. Section \ref{sec:Pricing} is dedicated to the applications to pricing and hedging a European option in complete linear and non-linear markets. In Subsection \ref{subsec:LinearPricing1}, we present an example of  a \textit{linear} and complete market; in Subsection \ref{subsec:non-linearPricingCompleteMarket}, we present the\textit{ non-linear } complete market with $p$ defaults, and study properties of the associated non-linear pricing system for a European option; in Subsection \ref{subsection_large1}, we provide a  particular case of a  non-linear complete market where a large seller influences  the probability of default of one (defaultable) risky asset and, in Subsection \ref{subsection_large2}, we  consider the example where the large seller's trading strategy influences the probabilities of default of all $p$  (defaultable) assets. 
\section{The Underlying Probability Setup}\label{sec:PPS}
In the sequel, we fix $T > 0$ to be the finite time horizon. Let $(\Omega, \mathcal{F}, P)$ be a complete probability space, and let $(W_t)_{t \in [0, T]}$ be a one-dimensional Brownian motion. Let $\mathbb{F} \coloneqq (\mathcal{F}_t)_{t \in [0, T]}$ denote the augmented filtration generated by $W$. Inequalities and equalities between random variables are to be understood in the almost sure (a.s.) sense with respect to $P$. 

Let $p \in \mathbb{N} \setminus \{0\}$, and let $\tau_1, \tau_2, \ldots, \tau_p$ be positive random times. We assume that, for each $i \in \{1, \ldots, p\}$, the random time $\tau_i$ has a continuous distribution, and that $P(\tau_i \ne \tau_j) = 1$ for all $i, j \in \{1, \ldots, p\}$ with $i \ne j$. Moreover, we assume that the $p$ default times are strictly ordered, i.e., $\tau_1 < \tau_2 < \cdots < \tau_p$. We will interpret  $\tau_i$  as the $i$-th default (or $i$-th credit event) time. For  $i \in \{1, \ldots, p\}$ and $t \in [0, T]$, we define $N_t^i \coloneqq \mathbbm{1}_{\{\tau_i \le t\}}$.

For each $i \in \{1, \ldots, p\}$, let $\mathbb{F}^i \coloneqq (\mathcal{F}_t^i)_{t \in [0, T]}$ denote the smallest right-continuous filtration making $\tau_i$ an $\mathbb{F}^i$-stopping time. We define the enlarged filtration $\mathbb{G} \coloneqq \mathbb{F} \vee \mathbb{F}^1 \vee \cdots \vee \mathbb{F}^p$. The filtration $\mathbb{G}$ is, in fact, the augmented filtration generated by $W$ and the default indicator processes $N^1, N^2, \ldots, N^p$. We \textit{assume} that the $\mathbb{F}$-Brownian motion $W$ remains a $\mathbb{G}$-Brownian motion.

By definition, for each $i$, the process $N^i = (N_t^i)_{t \in [0, T]}$ is non-decreasing, $\mathbb{G}$-adapted, and thus a $\mathbb{G}$-submartingale. Let $\Lambda^i = (\Lambda_t^i)$ denote the $\mathbb{G}$-predictable compensator of $N^i$. Note that the process $(\Lambda_{t \wedge \tau_i}^i)$ is the $\mathbb{G}$-predictable compensator of $(N_{t \wedge \tau_i}^i)$. By the uniqueness of the predictable compensator, we have $\Lambda_{t \wedge \tau_i}^i = \Lambda^i_t$ for $t \ge 0$ a.s.

We further assume that each process $\Lambda^i$ is absolutely continuous with respect to the Lebesgue measure. This implies the existence of a non-negative $\mathbb{G}$-predictable process $(\lambda_t^i)_{t \in [0, T]}$, called the $i$-th \emph{intensity process}, such that $
\Lambda_t^i = \int_0^t \lambda_s^i ds$, for $t \in [0,T]$.
Since $\Lambda_t^i = \Lambda_{t \wedge \tau_i}^i$, it follows that $\lambda_t^i = 0$ for $t > \tau_i$.

For each $i \in \{1, \ldots, p\}$, we define the $\mathbb{G}$-compensated default martingale $M^i$ by
\begin{equation}\label{eq:2.1}
    M_t^i = N_t^i - \int_0^t \lambda_s^i ds,
\end{equation}
for all $t \in [0, T]$.

\subsection{BSDEs with Multiple Default Jumps}\label{sub-section_BSDE}

We define the following:

\begin{itemize}
	\item $\mathcal{P}$ is the \textbf{predictable} $\sigma$-algebra on $\Omega\times[0,T]$.
    \item $\mathcal{S}_T^2$ is the set of $\mathbb{G}$-adapted right-continuous left-limited (rcll) processes $\varphi$ such that $\mathbb{E}[\sup_{0\le t\le T}|\varphi_t|^2]<+\infty$.
    \item $\mathcal{A}_T^2$ is the set of real-valued \textbf{finite variation} rcll $\mathbb{G}$-adapted processes $A$, with $\mathbb{E}[(\int_0^T|dA_t|)^2]<\infty$ and such that  $A_0=0$.
    \item $\mathcal{A}_{p, T}^2$ is the subset of all \textbf{predictable} processes in $\mathcal{A}_T^2$.
    \item $\mathcal{H}_T^2$ is the space of all $\mathbb{G}$-predictable processes endowed with  $\|Z\|^2 \coloneqq \mathbb{E}[\int_0^T|Z_t|^2dt]<\infty$.
    \item $\mathcal{H}_{\lambda^i, T}^2$ is the space $L^2(\Omega \times [0,T], \mathcal{P}, \lambda_tdP\otimes dt)$ equipped with the scalar product $\langle U,V\rangle_{\lambda^i}\coloneqq\mathbb{E}[\int_0^TU_tV_t\lambda_t^idt]<\infty$. For all $U\in\mathcal{H}_{\lambda^i,T}^2$ we define $\|U\|_{\lambda^i}^2\coloneqq\mathbb{E}[\int_0^T|U_t|^2\lambda_t^idt]<\infty$.
\end{itemize}

If it is obvious that we are working under the finite time horizon $T$, we might drop the $T$ subscript from the above notation.

As the $\mathbb{G}$-intensity $\lambda_t^i$ disappears for $t>\tau_i$ we have for all $U\in\mathcal{H}_{\lambda^i}^2$ $\|U\|_{\lambda^i}^2=\mathbb{E}[\int_0^{T\wedge\tau_i}|U_t|^2\lambda_t^idt]$ {(hence, the values of $U$ after  $\tau_i$ do not intervene in the computation of $\|U\|_{\lambda^i}^2$).

For our framework of multiple defaults, we recall the martingale representation property from \cite{kusuoka1999remark} (see also \cite{grigorian2023enlargement} Theorem 107).

\begin{theorem}[\textbf{Martingale Representation Property}]\label{theorem:MRP}
    For any $(\mathbb{P}, \mathbb{G})$-square-integrable martingale $(m_t)_{t\in[0,T]}$ there exist unique $\mathbb{G}$-predictable processes $z\in\mathcal{H}_T^2$ and  $k^i\in\mathcal{H}_{\lambda^i,T}^2$ for all $i\in\{1,\ldots,p\}$, such that the following martingale representation property holds,
    \begin{equation}\label{eq:2-MRP}
        m_t = m_0 + \int_0^tz_sdW_s + \sum_{i=1}^p\int_0^tk_s^idM_s^i.
    \end{equation}
 
\end{theorem}

In the following definition,  we extend the notion of  $\lambda$-driver from \cite{dumitrescu2018bsdes} in order to account for multiple default jumps. For simplicity we denote by $\lambda^{(p)}$ the vector $(\lambda^1,\ldots,\lambda^p)'$ of default intensities. Here the notation $'$ is for the transposition of the vector.

\begin{definition}[\textbf{$\lambda^{(p)}$-Admissible Driver}]\label{defintion:2_LambdaDriver}
    We say that a function $g$ is a driver if $g:\Omega\times[0,T]\times\mathbb{R}^{2+p}\rightarrow\mathbb{R}$ with $(\omega, t, y, z, k^1,\ldots,k^p)\mapsto g(\omega,t,y,z,k^1,\ldots,k^p)$ is a $\mathcal{P}\otimes\mathcal{B}(\mathbb{R}^{2+p})$-measurable function such that $g(\cdot,\cdot,0,0,\ldots,0)\in\mathcal{H}_T^2$. The driver $g$ is said to be $\lambda^{(p)}$-admissible if there exists a constant $C\ge0$ such that for $dP\otimes dt$-almost every $(\omega, t)$, for all $(y_1,z_1,k_1^1,\ldots,k_1^p)$, $(y_2, z_2, k_2^1,\ldots,k_2^p)$ we have,
    \begin{multline}\label{eq:2-LAD}
        |g(\omega, t, y_1,z_1,k_1^1,\ldots,k_1^p) - g(\omega, t, y_2,z_2,k_2^1,\ldots,k_2^p)| \le\\ C\left(|y_1-y_2|+|z_1-z_2| + \sum_{i=1}^p\sqrt{\lambda_t^i(\omega)}|k_1^i-k_2^i|\right).
    \end{multline}
\end{definition}

\begin{remark}
    From the condition  in \eqref{eq:2-LAD} we have that for each $(y, z, k^1,\ldots,k^p)$,  $g(t,y,z,k^1,\ldots,k^j,k^{j+1},\ldots,k^p)=g(t,y,z,0,\ldots,0,k^{j+1},\ldots,k^p)$ for $t>\tau_j$ $dP\otimes dt$-a.e., where we have used the fact that $\lambda^j$ disappears after $\tau_j$ and the assumption that the stopping times $\tau_1,\ldots,\tau_p$ are ordered. Hence, on the set $\{t>\tau_j\}$, the $\lambda^{(p)}$-admissible driver $g$ does not depend on $k^1,\ldots,k^j$.
\end{remark}

\begin{definition}[\textbf{BSDE with a $\lambda^{(p)}$-Admissible Driver}]
    Let $g$ be a $\lambda^{(p)}$-admissible driver and let $\eta\in L^2(\mathcal{G}_T)$. A process  $(Y, Z, K^1,\ldots,K^p)$ in $\mathcal{S}_T^2\times\mathcal{H}_T^2\times\mathcal{H}_{\lambda^1,T}^2\times\cdots\times\mathcal{H}_{\lambda^p,T}^2$ is said to be a solution of the BSDE with $p$ default jumps, with a terminal time $T$, a $\lambda^{(p)}$-admissible driver $g$, and a terminal condition $\eta$, if it satisfies the following:
    \begin{equation}
        -dY_t = g(t, Y_t, Z_t, K_t^1,\ldots,K_t^p)dt - Z_tdW_t - \sum_{i=1}^pK_t^idM_t^i,\quad Y_T=\eta.
    \end{equation}
\end{definition}

As we will see later, when dealing with (possibly non-linear) pricing problems in markets with defaults, contingent claims might  generate intermediate cash flows. These may arise, for instance,  from promised dividends, which can be modeled as a stream of continuous or discrete random cash flows received by the claim holder, or from a recovery process, which provides a recovery payoff in the event that a default occurs before time $T$. 
It is convenient to `wrap' these various sources of intermediate cash flows into a single `dividend' process $D$, where $D$ is assumed to be optional (but not necessarily predictable), right-continuous with left limits (rcll), and of finite variation.

Thus, we are interested in BSDEs with \textbf{generalized} drivers which include a process $D \in \mathcal{A}_T^2$.

\begin{definition}[\textbf{BSDE with a \emph{Generalized} $\lambda^{(p)}$-Admissible Driver}]\label{defintion:2_GenLambdaDriver}
    Let $g$ be a $\lambda^{(p)}$-admissible driver, let $\eta \in L^2(\mathcal{G}_T)$, and let $D \in \mathcal{A}_T^2$. A process $(Y, Z, K^1, \ldots, K^p)$ in $\mathcal{S}_T^2 \times \mathcal{H}_T^2 \times \mathcal{H}_{\lambda^1,T}^2 \times \cdots \times \mathcal{H}_{\lambda^p,T}^2$ is said to be a solution of the BSDE with $p$ default jumps, with a terminal time $T$, a \textbf{generalized} $\lambda^{(p)}$-admissible driver $g(t,y,z,k^1,\ldots,k^p)\,dt + dD_t$, and a terminal condition $\eta$, if it satisfies the following:
    \begin{equation}\label{eq:2gen_BSDE}
        -dY_t = g(t,Y_t, Z_t, K_t^1, \ldots, K_t^p)\,dt + dD_t - Z_t\,dW_t - \sum_{i=1}^p K_t^i\,dM_t^i, \quad Y_T = \eta.
    \end{equation}
\end{definition}

We emphasize that, in Equation \eqref{eq:2gen_BSDE}, the process $D$ is a finite variational, rcll, adapted process such that $D_0 = 0$ and its total variation is integrable. This implies that $D$ has at most a countable number of jumps and admits the canonical decomposition $D = A - A'$, where $A$ and $A'$ are integrable, non-decreasing, rcll, adapted processes starting at zero (i.e., $A_0 = A'_0 = 0$), and such that the mutual singularity condition $dA_t \perp dA'_t$ is satisfied (cf.,\ e.g., Proposition A.7 in \cite{dumitrescu2016generalized}). In the case where $D$ is predictable, the processes $A$ and $A'$ are also predictable.

\begin{proposition}\label{prop:OptionalpJumps}
    If $D\in\mathcal{A}_T^2$, then there exist a unique (predictable) process $D'\in\mathcal{A}_{p,T}^2$ and unique (predictable) processes $\theta^1\in\mathcal{H}_{\lambda^1,T}^2$, $\theta^2\in\mathcal{H}_{\lambda^2,T}^2$, $\ldots$ , $\theta^p\in\mathcal{H}_{\lambda^p,T}^2$, such that for all $t\in[0,T]$,
    \begin{equation}\label{eq:2_D_Optional}
        D_t = D_t' + \sum_{i=1}^p\int_0^t\theta_s^idN_s^i.
    \end{equation}
\end{proposition}

\begin{proof}

    As $D\in\mathcal{A}_T^2$, we have the canonical decomposition $D = A - \hat{A}$, where $A,\hat{A}$ are non-decreasing processes in $\mathcal{A}_T^2$, such that $dA_t\bot d\hat{A}$. By applying Lemma \ref{lemma:hDecomposition} to $A$ and $\hat{A}$, we get, that $A$ and $\hat{A}$ can be uniquely decomposed as,
    \begin{equation*}
        A_t = B_t + \sum_{i=1}^p\int_0^t \psi_s^idN_s^i,
    \end{equation*}
    and
    \begin{equation*}
        \hat{A}_t = \hat{B}_t + \sum_{i=1}^p\int_0^t \hat{\psi}_s^idN_s^i,
    \end{equation*}
    where $(B_t)$ and $(\hat{B}_t)$ are predictable (non-decreasing) and in $\mathcal{A}_{p,T}^2$, and for each $i\in\{1,\ldots,p\}$, $(\psi_t^i)$ and $(\hat{\psi}_t^i)$ are in $\mathcal{H}_{\lambda^i,T}^2$. By setting $D_t' \coloneqq B_t -\hat{B}_t$, and, for each $i\in\{1,\ldots,p\}$, $\theta_t^i\coloneqq\psi_t^i-\hat{\psi}_t^i$, we get the desired property.
\end{proof}

\subsection{BSDE with Multiple Default Jumps: Properties}\label{subsection_existence}

We begin by establishing \emph{a priori} estimates for BSDEs with $p$ default jumps. For $\beta > 0$, $\phi \in \mathcal{H}_T^2$, and $k^i \in \mathcal{H}_{\lambda^i,T}^2$, we introduce the following: $\|\phi\|_\beta^2 \coloneqq \mathbb{E}[\int_0^T e^{\beta t} \phi_t^2 \, dt]$, $\|k^1\|_{\lambda^1,\beta}^2 \coloneqq \mathbb{E}[\int_0^T e^{\beta t} (k_t^1)^2 \lambda_t^1 \, dt], \ldots, \|k^p\|_{\lambda^p,\beta}^2 \coloneqq \mathbb{E}[\int_0^T e^{\beta t} (k_t^p)^2 \lambda_t^p \, dt]$, and $\|k^{(p)}\|_{\lambda^{(p)},\beta}^2 \coloneqq \sum_{i=1}^p \|k^i\|_{\lambda^i,\beta}^2 = \mathbb{E}[\int_0^T e^{\beta t} \sum_{i=1}^p (k_t^i)^2 \lambda_t^i \, dt]$.

\subsubsection{A Priori Estimates for BSDEs with Multiple Default Jumps}

\begin{proposition}\label{prop:PE}
    Let $\eta$, $\hat{\eta}\in L^2(\mathcal{G}_T)$. Let $g$, $\hat{g}$ be two $\lambda^{(p)}$-admissible drivers. Let $C>0$ be a $\lambda^{(p)}$-constant associated to $g$. Let $D$ be an optional process belonging to $\mathcal{A}_T^2$.
    Let $(Y,Z,K^1,\ldots,K^p)$ and $(\hat{Y}, \hat{Z}, \hat{K}^1,\ldots,\hat{K}^p)$ be  solutions to the BSDEs associated with terminal time $T>0$, generalized drivers $g(t,y,z,k^1,\ldots,k^p)dt + dD_t$ and $\hat{g}(t,\hat{y},\hat{z},\hat{k}^1,\ldots,\hat{k}^p) + dD_t$ respectively, and terminal conditions $\eta$ and $\hat{\eta}$ respectively.
    Let $\Bar{\eta}\coloneqq\eta - \hat{\eta}$.  For $s\in[0,T]$, we denote $\Bar{Y}_s\coloneqq Y_s - \hat{Y}_s$, $\Bar{Z}_s\coloneqq Z_s - \hat{Z}_s$, and for each $i\in\{1,\ldots,p\}$, we denote $\Bar{K}_s^i\coloneqq K_s^i - \hat{K}_s^i$.

    Let $\xi$, $\beta>0$ be such that $\beta\ge\frac{p+2}{\xi}+2C$ and $\xi\le\frac{1}{C^2}$. Then, for each $t\in[0,T]$, it holds
    \begin{equation}\label{eq:peProp1}
        e^{\beta t}\Bar{Y}_t^2 \le \mathbb{E}\left[e^{\beta T}\Bar{\eta}^2\middle|\mathcal{G}_t\right] + \xi\text{ }\mathbb{E}\left[\int_t^Te^{\beta s}\Bar{g}_s^2ds\text{ }\middle|\text{ }\mathcal{G}_t\right]\quad\text{a.s.},
    \end{equation}
    where $\Bar{g}_s\coloneqq g(s,\hat{Y}_s,\hat{Z}_s,\hat{K}_s^1\ldots,\hat{K}_s^p) - \hat{g}(s,\hat{Y}_s,\hat{Z}_s,\hat{K}_s^1\ldots,\hat{K}_s^p)$. Further,
    \begin{equation}\label{eq:peProp2}
        \|\Bar{Y}\|_\beta^2\le T\left[e^{\beta T}\mathbb{E}[\Bar{\eta}^2] + \xi\text{ }\|\Bar{g}\|_\beta^2\right].
    \end{equation}
    Moreover, if $\xi < \frac{1}{C^2}$, we have,
    \begin{equation}\label{eq:peProp3}
        \|\Bar{Z}\|_\beta^2 + \|\Bar{K}^{(p)}\|_{\lambda^{(p)},\beta}^2 \le \frac{1}{(1-C^2\xi)}\left[e^{\beta T}\mathbb{E}[\Bar{\eta}^2] + \xi\text{ }\|\Bar{    g}\|_\beta^2\right].
    \end{equation}
\end{proposition}
\begin{proof}
    Using It\^o's formula applied to the semimartingale $(e^{\beta s}\Bar{Y}_s^2)$ between $t$ and $T$, we get,
    \begin{align}\label{eq:peProof1}
        \begin{aligned}
            e^{\beta T}\Bar{Y}_T^2 &= e^{\beta t}\Bar{Y}_t^2 + \beta\int_t^Te^{\beta s}\bar{Y}_{s-}^2ds + 2\int_t^Te^{\beta s}\Bar{Y}_{s-}d\bar{Y}_s \\&\quad+\int_t^Te^{\beta s}d\langle\bar{Y}^c\rangle_s + \sum_{t<s\le T}\left(e^{\beta s}\bar{Y}_s^2-e^{\beta s}\bar{Y}_{s-}^2 - 2e^{\beta s}\bar{Y}_{s-}\Delta\bar{Y}_s\right).
            \end{aligned}
    \end{align}
 	Further computation (noting that the $dD_t$ terms cancel) leads to,
    \begin{align}
        &\begin{aligned}\label{eq:peProof2}
        2\int_t^Te^{\beta s}\bar{Y}_{s-}d\bar{Y}_{s} &= -2\int_t^Te^{\beta s}\bar{Y}_{s-}(g(s,Y_s,Z_s,K_s^1,\ldots,K_s^p) - \hat{g}(s,\hat{Y}_s,\hat{Z}_s,\hat{K}_s^1,\ldots,\hat{K}_s^p))ds \\
        &\quad+ 2\int_t^Te^{\beta s}\bar{Y}_{s-}\bar{Z}_sdW_s + 2\int_t^Te^{\beta s}\bar{Y}_{s-}\sum_{i=1}^p\bar{K}_s^idM_s^i,
        \end{aligned}\\
        \label{eq:peProof3}
        &\quad\quad\int_t^Te^{\beta s}d\langle\bar{Y}^c\rangle_s = \int_t^Te^{\beta s}\bar{Z}_s^2ds,
    \end{align}
    and the `jump term',
    \begin{align*}
        \sum_{t<s\le T}\left(e^{\beta s}\bar{Y}_s^2-e^{\beta s}\bar{Y}_{s-}^2 - 2e^{\beta s}\bar{Y}_{s-}\Delta\bar{Y}_s\right) = \sum_{t<s\le T}e^{\beta s}(\bar{Y}_s - \bar{Y}_{s-})^2.
    \end{align*}
    Since $P(\tau_i = \tau_j) = 0$ for all $i,j\in\{1,\ldots,p\}$ such that $i\ne j$, we get,
    \begin{align}\label{eq:peProof4}
        \begin{aligned}
            &\sum_{t<s\le T}(e^{\beta s}\bar{Y}_s^2-e^{\beta s}\bar{Y}_{s-}^2 - 2e^{\beta s}\bar{Y}_{s-}\Delta\bar{Y}_s) = \sum_{t<s\le T}e^{\beta s}(\bar{Y}_s - \bar{Y}_{s-})^2\\
	      &= \int_t^Te^{\beta s}\sum_{i=1}^p(\bar{K}_s^i)^2dN_s^i\\
            &= \int_t^Te^{\beta s}\sum_{i=1}^p(\bar{K}_s^i)^2dM_s^i + \int_t^Te^{\beta s}\sum_{i=1}^p(\bar{K}_s^i)^2\lambda_s^ids.
        \end{aligned}
    \end{align}
    Plugging \eqref{eq:peProof2}, \eqref{eq:peProof3}, and \eqref{eq:peProof4}, into \eqref{eq:peProof1} we get,
    \begin{align}\label{eq:peProof5}
        \begin{aligned}
            e^{\beta t}\bar{Y}_t^2 &+ \beta\int_t^Te^{\beta s}\bar{Y}_{s-}^2ds + \int_t^Te^{\beta s}\left(\bar{Z}_s^2+\sum_{i=1}^p(\bar{K}_s^i)^2\lambda_s^i\right)ds = e^{\beta T}\bar{Y}_T^2\text{ }+\\ &2\int_t^Te^{\beta s}\bar{Y}_{s-}(g(s,Y_s,Z_s,K_s^1,\ldots,K_s^p) - \hat{g}(s,\hat{Y}_s,\hat{Z}_s,\hat{K}_s^1,\ldots,\hat{K}_s^p))ds\\
            &- 2\int_t^Te^{\beta s}\bar{Y}_{s-}\bar{Z}_sdW_s -2 \sum_{i=1}^p\int_t^Te^{\beta s}\bar{Y}_{s-}\bar{K}_s^idM_s^i-\sum_{i=1}^p\int_t^Te^{\beta s}(\bar{K}_s^i)^2dM_s^i.
        \end{aligned}
    \end{align}
    Taking the conditional expectation given $\mathcal{G}_t$ in \eqref{eq:peProof5} results in,
    \begin{align}\label{eq:interm99}
	\begin{aligned}
        e^{\beta t}\bar{Y}_t^2 &+ \mathbb{E}\left[\beta\int_t^Te^{\beta s}\bar{Y}_{s-}^2ds + \int_t^Te^{\beta s}\left(\bar{Z}_s^2+\sum_{i=1}^p(\bar{K}_s^i)^2\lambda_s^i\right)ds\middle|\mathcal{G}_t\right] \\&= \mathbb{E}\left[e^{\beta T}\bar{Y}_T^2\middle|\mathcal{G}_t\right] \\&+ 2\mathbb{E}\left[\int_t^Te^{\beta s}\bar{Y}_{s-}(g(s,Y_s,Z_s,K_s^1,\ldots,K_s^p) - \hat{g}(s,\hat{Y}_s,\hat{Z}_s,\hat{K}_s^1,\ldots,\hat{K}_s^p))ds\middle|\mathcal{G}_t\right].
	\end{aligned}
    \end{align}
    Now,
	\begin{multline*}
	g(s,Y_s,Z_s,K_s^1,\ldots,K_s^p) - \hat{g}(s,\hat{Y}_s,\hat{Z}_s,\hat{K}_s^1,\ldots,\hat{K}_s^p)\\= g(s,Y_s,Z_s,K_s^1,\ldots,K_s^p) - g(s,\hat{Y}_s,\hat{Z}_s,\hat{K}_s^1,\ldots,\hat{K}_s^p) + \bar{g}_s.
	\end{multline*}

Since $g$ is a $\lambda^{(p)}$-admissible driver, it satisfies condition \eqref{eq:2-LAD}; hence,
    \begin{multline*}
        |g(s,Y_s,Z_s,K_s^1,\ldots,K_s^p) - \hat{g}(s,\hat{Y}_s,\hat{Z}_s,\hat{K}_s^1,\ldots,\hat{K}_s^p)| \\\le C|\bar{Y}_s| + C|\bar{Z}_s| + C\sum_{i=1}^p|\bar{K}_s^i|\sqrt{\lambda_s^i} + |\bar{g}_s|.
    \end{multline*}

    For all $y,z, a, k^1,\lambda^1,\ldots,k^p,\lambda^p$ and $\epsilon>0$ we have the elementary inequalities,
    \begin{align*}
        2y\left(Cz + C\sum_{i=1}^pk^i\sqrt{\lambda^i} + a\right) &\le \frac{y^2}{\epsilon^2} + \epsilon^2\left(Cz+C\sum_{i=1}^pk^i\sqrt{\lambda^i} +a\right)^2\\
        &\le \frac{y^2}{\epsilon^2} + (p+2)\epsilon^2\left(C^2z^2 + C^2\sum_{i=1}^p(k^i)^2\lambda^i + a^2\right).
    \end{align*}

    Thus,
    \begin{align}\label{eq:peProof6}
        \begin{aligned}
            &\int_t^T e^{\beta s}2\bar{Y}_{s-}(g(s,Y_s,Z_s,K_s^1,\ldots,K_s^p) - \hat{g}(s,\hat{Y}_s,\hat{Z}_s,\hat{K}_s^1,\ldots,\hat{K}_s^p))ds \\&\le \left(2C+\frac{1}{\epsilon^2}\right)\int_t^Te^{\beta s}\bar{Y}_{s-}^2ds + (p+2)C^2\epsilon^2\int_t^Te^{\beta s}\left(\bar{Z}_s^2 + \sum_{i=1}^p(\bar{K}_s^i)^2\lambda_s^i\right)ds  \\&\quad+(p+2)\epsilon^2\int_t^Te^{\beta s}\bar{g}_s^2ds.
        \end{aligned}
    \end{align}
    Setting $\xi\coloneqq(p+2)\epsilon^2>0$, and using inequality \eqref{eq:peProof6} in \eqref{eq:interm99}, we have,
    \begin{align}\label{eq:peProof7}
        \begin{aligned}
            &e^{\beta t}\bar{Y}_t^2 \le \mathbb{E}[e^{\beta T}\bar{\eta}^2|\mathcal{G}_t] + \mathbb{E}\left[\left(2C + \frac{p+2}{\xi} - \beta\right)\int_t^Te^{\beta s}\bar{Y}_{s-}^2ds\middle|\mathcal{G}_t\right] +\\
            &\quad\mathbb{E}\left[(C^2\xi-1)\int_t^Te^{\beta s}\left(\bar{Z}_s^2+\sum_{i=1}^p(\bar{K}_s^i)^2\lambda_s^i\right)ds\middle|\mathcal{G}_t\right] + \mathbb{E}\left[\xi\int_t^Te^{\beta s}\bar{g}_s^2ds\middle|\mathcal{G}_t\right].
        \end{aligned}
    \end{align}
    Then for each $\xi,\beta>0$ such that $\beta \ge 2C + \frac{p+2}{\xi}$ and $\xi\le\frac{1}{C^2}$ we obtain the desired inequality  \eqref{eq:peProp1}.
    

    Using inequality \eqref{eq:peProof7}, integrating from $0$ to $T$, and taking the expectation, we get,
    \begin{align*}
        \|\bar{Y}\|_\beta^2&\le T\left[e^{\beta T}\mathbb{E}[\bar{\eta}^2]+\xi\|\bar{g}\|_\beta^2\right] + \mathbb{E}\left[\int_0^T\mathbb{E}\left[(C^2\xi - 1)\int_t^Te^{\beta s}(\bar{Z}_s^2+\sum_{i=1}^p(\bar{K}_s^i)^2)ds\middle|\mathcal{G}_t\right]dt\right]\\
        &\le T\left[e^{\beta T}\mathbb{E}[\bar{\eta}^2]+\xi\|\bar{g}\|_\beta^2\right] + (C^2\xi-1)\int_0^T\mathbb{E}\left[\int_0^Te^{\beta s}(\bar{Z}_s^2+\sum_{i=1}^p(\bar{K}_s^i)^2ds\right]dt\\
        &=T\left[e^{\beta T}\mathbb{E}[\bar{\eta}^2]+\xi\|\bar{g}\|_\beta^2\right] + T(C^2\xi - 1)(\|\bar{Z}\|_\beta^2+\|\bar{K}^{(p)}\|_{\lambda^{(p)},\beta}^2).
    \end{align*}
    By rearranging, we get
    \begin{align*}
        \|\bar{Y}\|_\beta^2 + T(1 - C^2\xi)(\|\bar{Z}\|_\beta^2+\|\bar{K}^{(p)}\|_{\lambda^{(p)},\beta}^2) \le T\left[e^{\beta T}\mathbb{E}[\bar{\eta}^2]+\xi\|\bar{g}\|_\beta^2\right] 
    \end{align*}
    Since $\xi\le\frac{1}{C^2}$,  we get  the inequality from \eqref{eq:peProp2}.\\
   Since $\|\bar{Y}\|_\beta^2\ge0$, we get,
    \begin{align*}
        \|\bar{Z}\|_\beta^2+\|\bar{K}^{(p)}\|_{\lambda^{(p)},\beta}^2 \le \frac{1}{1 - C^2\xi}\left[e^{\beta T}\mathbb{E}[\bar{\eta}^2]+\xi\|\bar{g}\|_\beta^2\right],
    \end{align*}
    which leads to the inequality \eqref{eq:peProp3} for $\xi<\frac{1}{C^2}$
    
\end{proof}

\begin{remark}\label{remark:PE}
   In the case of a $\lambda^{(p)}$-constant $C=0$, \eqref{eq:peProp1} and \eqref{eq:peProp2} hold for all $\xi,\beta>0$ such that $\beta\ge\frac{p+2}{\xi}$. Inequality \eqref{eq:peProp3} holds for all $\xi>0$ when $C=0$.
\end{remark}

\subsubsection{Existence and Uniqueness for BSDEs with Multiple Default Jumps}

With the \emph{a priori} estimates established for BSDEs with $p$ default jumps, we can now prove the existence and uniqueness of the solution. To do so, we make use of the representation property of square-integrable  $\mathbb{G}$-martingales (Theorem \ref{theorem:MRP}) and the \emph{a priori} estimates established in Proposition \ref{prop:PE}.\\
For $\beta>0$, we denote by $\mathcal{H}_\beta^{2,(p)}$ the space $\mathcal{S}^2\times\mathcal{H}_\beta^2\times\mathcal{H}_{\lambda^1,\beta}^2\times\cdots\times\mathcal{H}_{\lambda^p,\beta}^2$ equipped with the norm $\|\left(Y,Z,K^1,\ldots,K^p\right)\|_\beta^{2,(p)}\coloneqq\|Y\|_\beta^2+\|Z\|_\beta^2+\|K^1\|_{\lambda^1,\beta}^2+\cdots+\|K^p\|_{\lambda^p,\beta}^2$.

\begin{proposition}[\textbf{Exsitence and Uniqueness}]\label{prop:EU}
    Let $g$ be a $\lambda^{(p)}$-admissible driver, $\eta\in L^2(\mathcal{G}_T)$ and $D$ be an optional process in $\mathcal{A}_T^2$. There exists a unique solution $(Y,Z,K^1,\ldots,K^p)$ in $\mathcal{S}_T^2\times\mathcal{H}_T^2\times\mathcal{H}_{\lambda^1,T}^2\times\cdots\times\mathcal{H}_{\lambda^p,T}^2$ of the BSDE with multiple default jumps from Definition \ref{defintion:2_GenLambdaDriver}.
\end{proposition}
\begin{proof}
	The proof follows a standard two-step argument, where the second step relies on the \emph{a priori} estimates from Proposition \ref{prop:PE}.
	
	 We first consider the case where the driver is a driver process $g(t)$ which does not depend on $(y,z,k^1,\ldots,k^p)$. In this case,  the first component of the solution is, $Y_t = \mathbb{E}[\eta + \int_t^Tg(s)ds + D_T - D_t|\mathcal{G}_t]$. Applying the $\mathbb{G}$-martingale representation property to the square-integrable martingale $\mathbb{E}[\eta + \int_0^Tg(s)ds + D_T|\mathcal{G}_t]$ we get the processes $Z\in\mathcal{H}_T^2$ and $K^i\in\mathcal{H}_{\lambda^i,T}^2$ for $i\in\{1,\ldots,p\}$. These processes are unique due to the uniqueness in the $\mathbb{G}$-martingale representation result from Theorem \ref{theorem:MRP}. Hence there exists a unique solution to the BSDE with driver $g(s)ds + dD_s$, terminal time $T$ and terminal condition $\eta\in L^2(\mathcal{G}_T)$.
    
    We now focus on the case of a $\lambda^{(p)}$-admissible driver $g(t,y,z,k^1,\ldots,k^p)$.  We define a mapping $\mathbf{\Phi}$ from $\mathcal{H}_\beta^{2,(p)}$ to $\mathcal{H}_\beta^{2,(p)}$ as follows: for $(U,V,J^1,\ldots,J^p)\in\mathcal{H}_\beta^{2,(p)}$, $(Y,Z,K^1,\ldots,K^p)=\mathbf{\Phi}(U,V,J^1,\ldots,J^p)$ is the solution of the BSDE with driver $g(t,U_t,V_t,J_t^1,\ldots,J_t^p)dt + dD_t$, terminal time $T$ and terminal condition $\eta\in L^2(\mathcal{G}_T)$.
    The mapping is well-defined due to the first step of the proof.
    We  show that the mapping $\mathbf{\Phi}$ is a \textit{strict contraction}.  Let $(\hat{U},\hat{V},\hat{J}^1,\ldots,\hat{J}^p)\in\mathcal{H}_\beta^{2,p}$ and let $(\hat{Y},\hat{Z},\hat{K}^1,\ldots,\hat{K}^p)\coloneqq\mathbf{\Phi}(\hat{U},\hat{V},\hat{J}^1,\ldots,\hat{J}^p)$ be the solution of the BSDE with the driver $g(t,\hat{U}_t,\hat{V}_t,\hat{J}_t^1,\ldots,\hat{J}_t^p)dt + dD_t$, terminal time $T$ and terminal condition $\eta\in L^2(\mathcal{G}_T)$.

    We set $\bar{U}_t\coloneqq U_t-\hat{U}_t$, $\bar{V}_t\coloneqq V_t - \hat{V}_t$, $\bar{Y}_t\coloneqq Y_t - \hat{Y}_t$, $\bar{Z}_t\coloneqq Z_t - \hat{Z}_t$, and for each $i\in\{1,\ldots,p\}$ $\bar{J}_t^i\coloneqq J_t^i - \hat{J}_t^i$ and $\bar{K}_t^i\coloneqq K_t^i - \hat{K}_t^i$. We set $\Delta g_t\coloneqq g(t,U_t,V_t,J_t^1,\ldots,J_t^p) - g(t,\hat{U}_t,\hat{V}_t,\hat{J}_t^1,\ldots,\hat{J}_t^p)$.
    Then by the \emph{a priori} estimates from Proposition \ref{prop:PE} and Remark \ref{remark:PE}, applied to the driver processes $g_1(t)\coloneqq g(t,U_t,V_t,J_t^1,\ldots,J_t^p)$ and $g_2(t)\coloneqq g(t,\hat{U}_t,\hat{V}_t,\hat{J}_t^1,\ldots,\hat{J}_t^p)$ (where the driver $g_1(t)$ admits $C_1=0$ as a $\lambda^{(p)}$-constant since $g_1$ only depends on $(t,\omega)$), we have that,  for all $\xi,\beta>0$ such that $\beta\ge\frac{p+2}{\xi}$,
    \begin{align}\label{eq:EUproof1}
        \begin{aligned}
            \|\bar{Y}\|_\beta^2+\|\bar{Z}\|_\beta^2+\sum_{i=1}^p\|\bar{K}^i\|_{\lambda^i,\beta}^2 &\le \xi T\|\Delta g\|_\beta^2 + \xi \|\Delta g\|_\beta^2\\
            &=\xi(T+1)\|\Delta g\|_\beta^2.
        \end{aligned}
    \end{align}
 As by definition $g$ is a $\lambda^{(p)}$-admissible driver with $\lambda^{(p)}$ constant $C>0$ we get,
    \begin{align*}
        e^{\beta s}(\Delta g_s)^2 &\le e^{\beta s}C^2(|\bar{U}_s| + |\bar{V}_s| + \sum_{i=1}^p|\bar{J}_s^i|\sqrt{\lambda_s^i})^2\\
        &\le C^2(p+2)e^{\beta s}[\bar{U}_s^2+\bar{V}_s^2+\sum_{i=1}^p(\bar{J}_s^i)^2\lambda_s^i].
    \end{align*}
    Thus,
    \begin{align}\label{eq:EUproof2}
        \begin{aligned}
        \|\Delta g\|_\beta^2 &\le C^2(p+2)\mathbb{E}\left[\int_0^Te^{\beta s}(\bar{U}_s^2+\bar{V}_s^2+\sum_{i=1}^p(\bar{J}_s^i)^2\lambda_s^i)ds\right]\\
        &=C^2(p+2)(\|\bar{U}\|_\beta^2+\|\bar{V}\|_\beta^2+\sum_{i=1}^p\|\bar{J}^i\|_{\lambda^i,\beta}^2).
        \end{aligned}
    \end{align}
    Using inequalities \eqref{eq:EUproof2} and \eqref{eq:EUproof1}, we get,
    \begin{equation}\label{eq:EUproof3}
        \|\bar{Y}\|_\beta^2+\|\bar{Z}\|_\beta^2+\sum_{i=1}^p\|\bar{K}^i\|_{\lambda^i,\beta}^2 \le C^2(p+2)\xi(T+1)(\|\bar{U}\|_\beta^2+\|\bar{V}\|_\beta^2+\sum_{i=1}^p\|\bar{J}^i\|_{\lambda^i,\beta}^2),
    \end{equation}
    for all $\xi,\beta>0$ such that $\beta \ge \frac{p+2}{\xi}$. Choosing $\xi=\frac{1}{2(T+1)(p+2)C^2}$ and $\beta\ge2(p+2)^2(T+1)C^2$, we derive $\|\bar{Y},\bar{Z},\bar{K}^1,\ldots,\bar{K}^p\|_\beta^{2,(p)}\le\frac{1}{2}\|\bar{U},\bar{V},\bar{J}^1,\ldots,\bar{J}^p\|_\beta^{2,(p)}$.

    Hence for $\beta\ge2(p+2)^2(T+1)C^2$ we have that $\mathbf{\Phi}$ is a (strict) contraction from $\mathcal{H}_\beta^{2,(p)}$ to $\mathcal{H}_\beta^{2,(p)}$ and thus admits a unique fixed point $(Y,Z,K^1,\ldots,K^p)$ in the Banach space $\mathcal{H}_\beta^{2,(p)}$, which is the unique solution to the BSDE with driver $g(t,Y_t,Z_t,K_t^1,\ldots,K_t^p)dt + dD_t$, terminal time $T$ and terminal condition $\eta\in L^2(\mathcal{G}_T)$.
\end{proof}

\subsubsection{\texorpdfstring{Generalized $\lambda^{(p)}$-Linear BSDEs with Multiple Default Jumps}{lambda-p-Linear BSDEs with Multiple Default Jumps}}

We study the particular case of $\lambda^{(p)}$-\textbf{linear} BSDEs with multiple default jumps.
\begin{definition}[\textbf{$\lambda^{(p)}$-Linear Driver and Generalized $\lambda^{(p)}$-Linear Driver}]
    A driver $g$ is $\lambda^{(p)}$-linear if it is of the form
    \begin{equation}\label{eq:defLambdaLinear}
        g(t,y,z,k^1,\ldots,k^p)=\alpha_ty + \beta_tz + \sum_{i=1}^p\gamma_t^ik^i\lambda_t^i + \delta_t,
    \end{equation}
    where $\delta\coloneqq(\delta_t)_{t\in[0,T]}\in\mathcal{H}_T^2$ and $(\alpha_t)$,$(\beta_t)$ and $(\gamma_t^i)$ for $i\in\{1,\ldots,p\}$, are $\mathbb{R}$-valued predictable processes such that $(\alpha_t)$,$(\beta_t)$ and $(\gamma_t^i\sqrt{\lambda_t^i})$ for $i\in\{1,\ldots,p\}$, are bounded.
    For $D\in\mathcal{A}_T^2$ given, we define the \textbf{generalized} $\lambda^{(p)}$-linear driver as,
    \begin{equation}\label{eq:defLambdaLinear2}
        (\alpha_ty + \beta_tz + \sum_{i=1}^p\gamma_t^ik^i\lambda_t^i)dt + dD_t.
    \end{equation}
    
\end{definition}

\begin{remark}
    If $g$ is given by \eqref{eq:defLambdaLinear}, then using the transformation $\nu_t^i\coloneqq\gamma_t^i\sqrt{\lambda_t^i}$ for each $i\in\{1,\ldots,p\}$, we have that each $(\nu^i)$ is a bounded predictable process and,
    \begin{equation}
        g(t,y,z,k^1,\ldots,k^p) = \alpha_ty + \beta_tz + \sum_{i=1}^p\nu_t^ik^i\sqrt{\lambda_t^i} + \delta_t.
    \end{equation}
    Hence,  a $\lambda^{(p)}$-linear driver is also a $\lambda^{(p)}$-admissible driver.
\end{remark}

We are interested in finding explicitly the solution of a generalized $\lambda^{(p)}$-linear BSDE. To do so, we first need a preliminary result on exponential local martingales in our framework.

\begin{remark}\label{remark:2.12}
 Let $\Gamma\coloneqq(\Gamma_t)_{t\in[0,T]}$ be the process satisfying the SDE,
    \begin{equation}\label{eq:defESpDefault}
        d\Gamma_s = \Gamma_{s-}(\beta_sdW_s + \sum_{i=1}^p\gamma_s^idM_s^i),\quad\Gamma_0=1.
    \end{equation}

    From Lemma \ref{lemma:expSemiMartingale}, we have: for all $s\ge0$,
    \begin{equation}\label{eq:defESpDefault2}
        \Gamma_s = \exp\left(\int_0^s\beta_rdW_r - \frac{1}{2}\int_0^s\beta_r^2dr\right)\exp\left(-\int_0^s\sum_{i=1}^p\gamma_r^i\lambda_r^idr\right)\prod_{i=1}^p(1+\gamma_{\tau_i}^i\mathbbm{1}_{\{s\ge\tau_i\}}),\text{  a.s.}
    \end{equation}
    If for all $i\in\{1,\ldots,p\}$  $\gamma_{\tau_i}^i\ge-1$ (respectively $\gamma_{\tau_i}^i>-1$) a.s., then $\Gamma_s\ge0$ (respectively $\Gamma_s>0$) for all $s\ge0$ a.s.
\end{remark}

\begin{proposition}\label{prop:ESM-M}
    Let $T>0$. If the random variable $\int_0^T(\beta_s^2+\sum_{i=1}^p(\gamma_s^i)^2\lambda_s^i)ds$ is bounded, then the exponential local martingale $(\Gamma_t)_{ t\in[0, T]}$, defined by \eqref{eq:defESpDefault}, is a martingale and satisfies $\mathbb{E}[\sup_{0\le t\le T}\Gamma_t^2]<+\infty$.
\end{proposition}
\begin{proof}
    From \eqref{eq:defESpDefault} the process $\Gamma$ is a local martingale. We  show that $\mathbb{E}[\sup_{0\le t\le T}\Gamma_t^2]<\infty$. 
    Let $dX_t = \beta_tdW_t + \sum_{i=1}^p\gamma_t^idM_t^i$. We have \footnote{For this, we  use  $\Delta \Gamma_t = \Gamma_t-\Gamma_{t-}$  and $\Gamma_t = \Gamma_{t-} + \Gamma_{t-}\Delta X_t$ ((by \eqref{eq:defESpDefault2}); we get $\Delta \Gamma_t = \Gamma_{t-}\Delta X_t$.} $d\Gamma_t = \Gamma_{t-}dX_t$ and $\Delta \Gamma_t = \Gamma_{t-}\Delta X_t$. Using this, we get, 
    \begin{align*}
        d[\Gamma]_t &= d\langle\Gamma^c\rangle_t + d\left(\sum_{0<s\le t}(\Delta \Gamma_s)^2\right)
        =\Gamma_{t-}^2d\langle X^c\rangle_t + d\left(\sum_{0< s\le t}\Gamma_{s-}^2(\Delta X_s)^2\right)\\
        &= \Gamma_{t-}^2\beta_t^2dt + d\left(\sum_{i=1}^p\int_0^t\Gamma_{s-}^2(\gamma_s^i)^2dN_s^i\right)= \Gamma_{t-}^2\beta_t^2dt + \Gamma_{t-}^2\sum_{i=1}^p(\gamma_t^i)^2dN_t^i.
    \end{align*}
    Using It\^o's formula applied to $(\Gamma_t^2)$ and  the fact that  $dN_t^i = dM_t^i + \lambda_t^idt$, we get,
    \begin{align}\label{eq:ExpSemiProof1}
        \begin{aligned}
            &d\Gamma_t^2 =2\Gamma_{t-}d\Gamma_t + d[\Gamma]_t= \Gamma_{t-}^2(2\beta_tdW_t + 2\sum_{i=1}^p\gamma_t^idM_t^i + \beta_t^2dt + \sum_{i=1}^p(\gamma_t^i)^2dN_t^i)\\
            &=\Gamma_{t-}^2\left[2\beta_tdW_t + \sum_{i=1}^p(2\gamma_t^i + (\gamma_t^i)^2)dM_t^i + \left(\beta_t^2+\sum_{i=1}^p(\gamma_t^i)^2\lambda_t^i\right)dt\right].
        \end{aligned}
    \end{align}
 This can be written in the form $d\Gamma_t^2 = \Gamma_{t-}^2dY_t$, where,
\begin{equation*}
	dY_t \coloneqq \left(\beta_t^2+\sum_{i=1}^p(\gamma_t^i)^2\lambda_t^i\right)dt + 2\beta_tdW_t + \sum_{i=1}^p(2\gamma_t^i + (\gamma_t^i)^2)dM_t^i.
\end{equation*}
We have  $dY_t\coloneqq dY_t^{(1)}+dY_t^{(2)}$, where $Y_t^{(1)}\coloneqq\int_0^t\left(\beta_s^2+\sum_{i=1}^p(\gamma_s^i)^2\lambda_s^i\right)ds$ and $Y_t^{(2)}\coloneqq\int_0^t2\beta_sdW_s+\int_0^t\sum_{i=1}^p(2\gamma_s^i + (\gamma_s^i)^2)dM_s^i$. We have, 
    \begin{equation}
        \mathcal{E}(Y^{(1)})_t = \exp\left(\int_0^t\left(\beta_s^2+\sum_{i=1}^p(\gamma_s^i)^2\lambda_s^i\right)ds\right).
    \end{equation}
    Using Lemma \ref{lemma:expSemiMartingale}, applied to $Y^{(2)}$, we get,
    \begin{multline}
        \mathcal{E}(Y^{(2)})_t = \exp\left(\int_0^t2\beta_sdW_s - \int_0^t\sum_{i=1}^p(2\gamma_s^i+(\gamma_s^i)^2)\lambda_s^ids - \int_0^t2\beta_s^2ds\right)\\\times\prod_{i=1}^p\left(1+(2\gamma_{\tau_i}^i + (\gamma_{\tau_i}^i)^2)\mathbbm{1}_{\{\tau_i\le t\}}\right).
    \end{multline}

    Using the identity $\mathcal{E}(Y^{(1)} + Y^{(2)} + [Y^{(1)},Y^{(2)}])_t = \mathcal{E}(Y^{(1)})_t\mathcal{E}(Y^{(2)})_t$ and the fact that $[Y^{(1)},Y^{(2)}]_t=0$ for all $t$ a.s. we get,
    \begin{equation}\label{eq:PropositionEXM1}
        \begin{aligned}
            \mathcal{E}(Y^{(1)})_t\mathcal{E}(Y^{(2)})_t = \Gamma_t^2&= \exp\left(\int_0^t\left(\beta_s^2+\sum_{i=1}^p(\gamma_s^i)^2\lambda_s^i\right)ds\right)\\
            &\times\exp\left(\int_0^t2\beta_sdW_s - \int_0^t\sum_{i=1}^p(2\gamma_s^i+(\gamma_s^i)^2)\lambda_s^ids - \int_0^t2\beta_s^2ds\right)\\
            &\times\prod_{i=1}^p\left(1+(2\gamma_{\tau_i}^i + (\gamma_{\tau_i}^i)^2)\mathbbm{1}_{\{\tau_i\le t\}}\right).
        \end{aligned}
    \end{equation}
Setting $\zeta_t\coloneqq\mathcal{E}(Y^{(2)})_t$, we have that $\zeta$ is an exponential local martingale with dynamics, $d\zeta_t = \zeta_{t-}dY_t^{(2)}$; more specifically,  
    \begin{equation}\label{eq:PropositionEXM1.5}
        d\zeta_t = \zeta_{t-}\left[2\beta_tdW_t + \sum_{i=1}^p(2\gamma_t^i + (\gamma_t^i)^2)dM_t^i\right],\quad\zeta_0=1.
    \end{equation}
    
    Thus, the exponential local martingale $\Gamma^2$ from \eqref{eq:PropositionEXM1} becomes,
    \begin{equation}\label{eq:PropositionEXM2}
        \Gamma_t^2 = \zeta_t\exp\left(\int_0^t\left(\beta_s^2+\sum_{i=1}^p(\gamma_s^i)^2\lambda_s^i\right)ds\right).
    \end{equation}
    By \eqref{eq:PropositionEXM1.5}, the local martingale $\zeta$ is non-negative. This implies that $\zeta$ is a supermartingale and hence $\mathbb{E}[\zeta_T] \le 1$. Now by the assumption that $\int_0^T(\beta_t^2 + \sum_{i=1}^p(\gamma_t^i)^2\lambda_t^i)dt$ is bounded, we get,
    \begin{equation*}
        \mathbb{E}[\Gamma_T^2]\le\mathbb{E}[\zeta_T]K \le K,
    \end{equation*}
    where $K>0$ is a constant depending on $\int_0^T(\beta_t^2 + \sum_{i=1}^p(\gamma_t^i)^2\lambda_t^i)dt$.
    By martingale inequalities, we get $\mathbb{E}[\sup_{0\le t \le T}\Gamma_t^2]<\infty$. We conclude that $\Gamma$ is a martingale.
\end{proof}

We now establish the explicit form of the (first component of the solution) solution of the  BSDE with a generalized $\lambda^{(p)}$-linear driver. We begin with the case where the finite variational process $D$ is \textbf{predictable}.

\begin{theorem}[\textbf{Explicit Solution of the Generalized $\lambda^{(p)}$-Linear BSDE with $D$ Predictable}]\label{theorem:RRGLB}

    Let $(\alpha_t)$, $(\beta_t)$ and $(\gamma_t^i)$, for  $i\in\{1,\ldots,p\}$, be $\mathbb{R}$-valued predictable processes such that $(\alpha_t)$, $(\beta_t)$ and $(\gamma_t^i\sqrt{\lambda_t^i})$, for $i\in\{1,\ldots,p\}$, are bounded.
    Let $\eta\in L^2(\mathcal{G}_T)$ and let $D$ be a (predictable) process  in  $\mathcal{A}_{p,T}^2$.
    Let $(Y,Z,K^1,\ldots,K^p)$ be the solution in $\mathcal{S}^2\times\mathcal{H}^2\times\mathcal{H}_{\lambda^1}^2\times\cdots\times\mathcal{H}_{\lambda^p}^2$ of the following  BSDE  with the generalized $\lambda^{(p)}$-linear driver $(\alpha_ty + \beta_tz+\sum_{i=1}^p\gamma_t^ik^i\lambda_t^i)dt + dD_t$, terminal time $T$ and terminal condition $\eta$,
    \begin{equation}\label{eq:ThmDpred1}
        -dY_t = \left(\alpha_tY_t + \beta_tZ_t + \sum_{i=1}^p\gamma_t^iK_t^i\lambda_t^i\right)dt + dD_t - Z_tdW_t - \sum_{i=1}^pK_t^idM_t^i,\quad Y_T = \eta.
    \end{equation}

    For each $t\in[0,T]$, let $(\Gamma_{t,s})_{s\ge t}$ be the unique solution of the following \textbf{adjoint forward SDE},
    \begin{equation}\label{eq:ThmDpred2}
        d\Gamma_{t,s} = \Gamma_{t,s-}\left(\alpha_sds + \beta_sdW_s + \sum_{i=1}^p\gamma_s^idM_s^i\right),\quad\Gamma_{t,t}=1.
    \end{equation}
    
    Then, the process $(Y_t)$ has the explicit form:
    \begin{equation}\label{eq:ThmDpred3}
        Y_t = \mathbb{E}\left[\Gamma_{t,T}\eta + \int_t^T\Gamma_{t,s-}dD_s\middle|\mathcal{G}_t\right],\quad0\le t\le T,\quad\text{a.s.}
    \end{equation}
\end{theorem}
\begin{remark}\label{remark:thmRR_Dpred}
    Applying Lemma \ref{lemma:expSemiMartingale} to the process $(\Gamma_{t,s})_{s\ge t}$ gives that $(\Gamma_{t,s})_{s\ge t}$ satisfies,
    \begin{align}\label{eq:remarkAdjoint1}
        \begin{aligned}
            \Gamma_{t,s} &= \exp\left(\int_t^s\alpha_rdr\right)\mathcal{E}\left(\int_t^.\beta_rdW_r + \int_t^.\sum_{i=1}^p\gamma_r^idM_r\right)_s\\
            &= \exp\left(\int_t^s\alpha_rdr+\int_t^s\beta_rdW_r - \frac{1}{2}\int_t^s\beta_r^2dr - \sum_{i=1}^p\gamma_r^i\lambda_r^idr\right)\prod_{i=1}^p(1+\gamma_{\tau_i}^i\mathbbm{1}_{\{t<\tau_i\le s\}}),
        \end{aligned}
    \end{align}
    for all $ t \le s\le T$ a.s. The process $(e^{\int_t^s\alpha_rdr})_{t\le s\le T}$ is positive and \textbf{bounded} (as $\alpha$ is bounded), and using Proposition \ref{prop:ESM-M} (since $\beta$ and $\gamma^i\sqrt{\lambda^i}$ for each $i\in\{1,\ldots,p\}$ are bounded), we have that $(\Gamma_{t,s})_{t\le s\le T}$ is a martingale and satisfies $\mathbb{E}[\sup_{t\le s\le T}\Gamma_{t,s}^2]<+\infty$.
\end{remark}
\begin{proof}
    Fix $t\in[0,T]$. Since $D\in\mathcal{A}_{p,T}^2$ is \emph{predictable} and the process $\Gamma_{t,.}$ admits at most $p$ jumps and only at the \emph{totally inaccessible} times $\tau_1,\ldots,\tau_p$, we have that $[\Gamma_{t,.},D]_s=0$ for $s\ge t$ a.s. By applying It\^o's product rule to $(Y_s\Gamma_{t,s})$, we get,
    \begin{equation}\label{eq:proofThmDpred1}
        -d(Y_s\Gamma_{t,s}) = -Y_{s-}d\Gamma_{t,s} - \Gamma_{t,s-}dY_s - d[Y,\Gamma_{t,.}]_s.
    \end{equation}
    Moreover,
    \begin{align}\label{eq:proofThmDpred2}
        \begin{aligned}
            d[Y,\Gamma_{t,.}]_s &= d\left[\int_t^\cdot\left(\alpha_rY_r+\beta_rZ_r+\sum_{i=1}^p\gamma_r^iK_r^i\lambda_r^i\right)dr,\Gamma_{t,\cdot}\right]_s + d[D,\Gamma_{t,\cdot}]_s \\
            &\quad\quad+ d[Z\bullet W, \Gamma_{t,\cdot}]_s + d\left(\sum_{i=1}^p[K^i\bullet M^i,\Gamma_{t,\cdot}]_s\right)\\
            &=d[D,\Gamma_{t,\cdot}]_s + d\left[\int_t^\cdot Z_rdW_r, \int_t^\cdot\Gamma_{t,r-}\beta_rdW_r\right]_s\\
	      &\quad\quad+ d\left(\sum_{i=1}^p\sum_{j=1}^p\left[\int_t^\cdot K_r^idM_r^i,\int_t^\cdot\Gamma_{t,r-}\gamma_r^jdM_r^j\right]_s\right)\\
            &= d\left(\int_t^s\Gamma_{t,r-}\beta_rZ_rdr\right) + d\left(\int_t^s\sum_{i=1}^p\sum_{j=1}^pK_r^i\Gamma_{t,r-}\gamma_r^jd[M^i,M^j]_r\right)\\
            &= \Gamma_{t,s-}\beta_sZ_sds + d\left(\int_t^s\sum_{i=1}^p\Gamma_{t,r-}K_r^i\gamma_r^idN_r^i\right)\\
            &= \Gamma_{t,s-}\beta_sZ_sds + \Gamma_{t,s-}\sum_{i=1}^pK_s^i\gamma_s^idN_s^i,
        \end{aligned}
    \end{align}
    where we have used that $d[M^i,M^j]_s = 0$, for $i\ne j$, since $P(\tau_i = \tau_j)=0$, $i\ne j$, and, for the case $i=j$, $d[M^i]_s = dN_s^i$.
    
    Plugging \eqref{eq:proofThmDpred2} into \eqref{eq:proofThmDpred1} and using $dN_s^i = dM_s^i + \lambda_s^ids$, we get,
    \begin{equation}\label{eq:proofThmDpred3}
        -d(Y_s\Gamma_{t,s}) = -\Gamma_{t,s-}(Y_s\beta_s + Z_s)dW_s - \Gamma_{t,s-}\left(\sum_{i=1}^p (Y_{s-}\gamma_s^i + K_s^i(1+\gamma_s^i))dM_s^i\right)\\
        +\Gamma_{t,s-}dD_s.
    \end{equation}
    Setting $dm_s = \Gamma_{t,s-}(Y_s\beta_s + Z_s)dW_s + \Gamma_{t,s-}\left(\sum_{i=1}^p (Y_{s-}\gamma_s^i + K_s^i(1+\gamma_s^i))dM_s^i\right)$, we get $-d(Y_s\Gamma_{t,s}) = -dm_s + \Gamma_{t,s-}dD_s$. Integrating between $t$ and $T$, we derive,
    \begin{equation}\label{eq:proofThmDpred4}
        Y_t = \eta\Gamma_{t,T} + \int_t^T\Gamma_{t,s-}dD_s - (m_T - m_t),\quad\text{a.s.}
    \end{equation}
    By Remark \ref{remark:thmRR_Dpred} we have $(\Gamma_{t,s})_{t\le s\le T}\in\mathcal{S}^2$. Furthermore, $Y\in\mathcal{S}^2$, $Z\in\mathcal{H}^2$ and $K^i\in\mathcal{H}_{\lambda^i}^2$ for each $i\in\{1,\ldots,p\}$, and $\beta$ and $\gamma^i\sqrt{\lambda^i}$ for each $i\in\{1,\ldots,p\}$ are bounded. It follows that the local martingale $m = (m_s)_{t\le s\le T}$ is a martingale.  Taking the conditional expectation in \eqref{eq:proofThmDpred4}, we get the desired equality \eqref{eq:ThmDpred3}.
\end{proof}

We now consider the case where the process $D$ just an \textbf{optional} process (not necessarily predictable). More precisely, $D$ is in $\mathcal{A}_{T}^2$; hence, by Proposition \ref{prop:OptionalpJumps}, it is of the form \eqref{eq:2_D_Optional}.

\begin{theorem}[\textbf{Explicit Solution of the Generalized $\lambda^{(p)}$-Linear BSDE with $D$ Optional}]\label{theorem:RRGLB2}
    Let the assumptions made in Theorem \ref{theorem:RRGLB} all hold, except that $D$ is now in $\mathcal{A}_T^2$ (and not necessarily in $\mathcal{A}_{p,T}^2$). Let $D'\in\mathcal{A}_{p,T}^2$ and $\theta^i\in\mathcal{H}_{\lambda^i,T}^2$, for $i\in\{1,\ldots,p\}$, be the unique predictable processes from Proposition \ref{prop:OptionalpJumps}, such that for all $t\in[0,T]$,
    \begin{equation}\label{eq:ThmDOptRR1}
        D_t = D_t' + \int_0^t\sum_{i=1}^p\theta_s^idN_s^i, \text{  a.s.}
    \end{equation}

    Let $(Y,Z,K^1,\ldots,K^p)$ be the solution in $\mathcal{S}^2\times\mathcal{H}^2\times\mathcal{H}_{\lambda^1}^2\times\cdots\times\mathcal{H}_{\lambda^p}^2$ of the BSDE  with generalized $\lambda^{(p)}$-linear driver $(\alpha_ty+\beta_tz+\sum_{i=1}^p\gamma_t^ik^i\lambda_t^i)dt + dD_t$, terminal time $T$, and terminal condition $\eta\in L^2({\mathcal{G}_T})$.

    Then, a.s. for all $t\in[0,T]$,
    \begin{align}\label{eq:ThmDOptRR2}
        \begin{aligned}
            Y_t &= \mathbb{E}\left[\Gamma_{t,T}\eta + \int_t^T\Gamma_{t,s-}\left(dD_s' + \sum_{i=1}^p\theta_s^i(1+\gamma_s^i)dN_s^i\right)\middle|\mathcal{G}_t\right]\\
            &=\mathbb{E}\left[\Gamma_{t,T}\eta+\int_t^T\Gamma_{t,s-}dD_s' + \sum_{i=1}^p\Gamma_{t,\tau_i}\theta_{\tau_i}^i\mathbbm{1}_{\{t < \tau_i \le T\}}\middle|\mathcal{G}_t\right],
        \end{aligned}
    \end{align}
    where the process $(\Gamma_{t,s})_{t\le s\le T}$ is the solution of  the adjoint forward SDE \eqref{eq:remarkAdjoint1}.
\end{theorem}
\begin{proof}
    Since $D$ satisfies \eqref{eq:ThmDOptRR1}, we have,
    \begin{align}\label{eq:proofThmDOptRR1}
        \begin{aligned}
            d[D,\Gamma_{t,\cdot}]_s &= d[D',\Gamma_{t,\cdot}]_s + d\left(\sum_{i=1}^p\sum_{j=1}^p\left[\int_t^\cdot\theta_r^idN_r^i, \int_t^\cdot\Gamma_{t,r-}\gamma_r^jdN_r^j\right]_s\right)\\
            &=d\left(\int_t^s\Gamma_{t,r-}\sum_{i=1}^p\theta_r^i\gamma_r^idN_r^i\right)=\Gamma_{t,s-}\sum_{i=1}^p\theta_s^i\gamma_s^idN_s^i\quad\text{a.s.},
        \end{aligned}
    \end{align}
    where we have used that $\tau_i\neq \tau_j$ a.s. for $i\neq j.$ 
     By applying It\^o's product rule to $(Y_s\Gamma_{t,s})$, using similar computations to those from the proof of Theorem \ref{theorem:RRGLB}, and using \eqref{eq:proofThmDOptRR1}, we get,
    \begin{multline}\label{eq:proofThmDOptRR2}
        -d(Y_s\Gamma_{t,s}) = -\Gamma_{t,s-}(Y_s\beta_s + Z_s)dW_s - \Gamma_{t,s-}\left(\sum_{i=1}^p(Y_{s-}\gamma_s^i + K_s^i(1+\gamma_s^i))dM_s^i\right)\\
        -\Gamma_{t,s-}\left(dD_s + \sum_{i=1}^p\theta_s^i\gamma_s^idN_s^i\right).
    \end{multline}

    Using $\Gamma_{t,s-}(dD_s + \sum_{i=1}^p\theta_s^i\gamma_s^idN_s^i) = \Gamma_{t,s-}(dD_s' + \sum_{i=1}^p\theta_s^i(1+\gamma_s^i)dN_s^i)$ in \eqref{eq:proofThmDOptRR2},  integrating from $t$ to $T$ and taking the conditional expectation, we derive that,
    \begin{equation}\label{eq:proofThmDOptRR3}
        Y_t = \mathbb{E}\left[\Gamma_{t,T}\eta + \int_t^T\Gamma_{t,s-}\left(dD_s' + \sum_{i=1}^p\theta_s^i(1+\gamma_s^i)dN_s^i\right)\middle|\mathcal{G}_t\right],
    \end{equation}
    which is the first equality from \eqref{eq:ThmDOptRR2}. 

    Now, we have,
    \begin{align}\label{eq:proofThmDOptRR4}
        \begin{aligned}
            \mathbb{E}\left[\int_t^T\Gamma_{t,s-}\sum_{i=1}^p\theta_s^i(1+\gamma_s^i)dN_s^i\middle|\mathcal{G}_t\right] &= \mathbb{E}\left[\sum_{i=1}^p\Gamma_{t,\tau_i-}\theta_{\tau_i}^i(1+\gamma_{\tau_i}^i)\mathbbm{1}_{\{t<\tau_i\le T\}}\middle|\mathcal{G}_t\right]\\
            &=\mathbb{E}\left[\sum_{i=1}^p\Gamma_{t,\tau_i}\theta_{\tau_i}^i\mathbbm{1}_{\{t<\tau_i\le T\}}\middle|\mathcal{G}_t\right].
        \end{aligned}
    \end{align}

    The second equality in \eqref{eq:proofThmDOptRR4} is due to $\Gamma_{t,s}$ having the following representation,
    \begin{align*}
        \Gamma_{t,s} = \exp\left(\int_t^s\left(\alpha_r-\frac{1}{2}\beta_r^2-\sum_{i=1}^p\gamma_r^i\lambda_r^i\right)dr+\int_t^s\beta_rdW_r\right)\prod_{i=1}^p(1+\gamma_{\tau_i}^i\mathbbm{1}_{\{t<\tau_i\le s\}}),
    \end{align*}
 	for all $t\le s\le T$ a.s.;  it follows that, for each $i\in\{1,\ldots,p\},$ $\Gamma_{t,\tau_i-}(1+\gamma_{\tau_i}^i)\mathbbm{1}_{\{t<\tau_i\le T\}} = \Gamma_{t,\tau_i}\mathbbm{1}_{\{t<\tau_i\le T\}}$ (where we have used that $\tau_1<\tau_2<...<\tau_p$).  

    By replacing \eqref{eq:proofThmDOptRR4} in \eqref{eq:proofThmDOptRR3}, we get the following representation,
    \begin{equation*}
        Y_t = \mathbb{E}\left[\Gamma_{t,T}\eta + \int_t^T\Gamma_{t,s-}dD_s' + \sum_{i=1}^p\Gamma_{t,\tau_i}\theta_{\tau_i}^i\mathbbm{1}_{\{t<\tau_i\le T\}}\middle|\mathcal{G}_t\right],
    \end{equation*}
    which is the second equality in \eqref{eq:ThmDOptRR2}.
\end{proof}

\subsection{Comparison Theorems for BSDEs with Multiple Default Jumps}\label{subsection_comparison}

We now provide a comparison and a strict comparison results for BSDEs with generalized $\lambda^{(p)}$-admissible drivers associated with finite variational rcll adapted processes in $\mathcal{A}_T^2$.  \\
For convenience, we define the following sets:
\begin{equation} \label{sets}
\begin{aligned}
&A_0:=\{\tau_1>T\}, A_1:=\{\tau_1\leq T,\tau_2>T\},..., A_k:=\{\tau_k\leq T,\tau_{k+1}> T\},...,\\
&A_{p-1}:=\{\tau_{p-1}\leq T,\tau_{p}> T\}, \text { and } A_p:=\{\tau_{p}\leq T\}.
\end{aligned}
\end{equation}
 As $\tau_1<\tau_2<...<\tau_p$, the above sets  form a partition of $\Omega$.

\begin{theorem}[\textbf{Comparison and Strict Comparison for BSDEs with Multiple Default Jumps}]\label{thm:CompThm}
    Let $\eta$ and $\hat{\eta}$ be in $L^2(\mathcal{G}_T)$. Let $g$ and $\hat{g}$ be $\lambda^{(p)}$-admissible drivers. Let $D$ and $\hat{D}$ be  processes in $\mathcal{A}_{T}^2$. Let $(Y,Z,K^1,\ldots,K^p)$ be the solution in $\mathcal{S}^2\times\mathcal{H}_T^2\times\mathcal{H}_{\lambda^1,T}^2\times\cdots\times\mathcal{H}_{\lambda^p,T}^2$ to the BSDE,
    \begin{equation*}
        -dY_t = g(t,Y_t,Z_t,K_t^1,\ldots,K_t^p)dt + dD_t - Z_tdW_t - \sum_{i=1}^pK_t^idM_t^i,\quad Y_T =\eta.
    \end{equation*}
    Let $(\hat{Y},\hat{Z},\hat{K}^1,\ldots,\hat{K}^p)$ be the solution in $\mathcal{S}^2\times\mathcal{H}_T^2\times\mathcal{H}_{\lambda^1,T}^2\times\cdots\times\mathcal{H}_{\lambda^p,T}^2$ to the BSDE,
    \begin{equation*}
        -d\hat{Y}_t = \hat{g}(t,\hat{Y}_t,\hat{Z}_t,\hat{K}_t^1,\ldots,\hat{K}_t^p)dt + d\hat{D}_t - \hat{Z}_tdW_t - \sum_{i=1}^p\hat{K}_t^idM_t^i,\quad \hat{Y}_T =\hat{\eta}.
    \end{equation*}
    Then, the following two statements hold true:
    \begin{enumerate}[label=(\roman*)]
        \item \textbf{Comparison:} Assume that there exist $p$ predictable processes $(\gamma_t^i)$ (where $i\in\{1,\ldots,p\}$) with      $(\gamma_t^i\sqrt{\lambda_t^i})\text{  bounded }dP\otimes dt$ \text{-a.e.} (for  $i\in\{1,\ldots,p\}$) such that,
        \begin{equation}\label{eq:CompThm1}
       \text{ for each  } k\in\{1,...,p\},  \text{ on } A_k, 1+\gamma_{\tau_i}^i\ge0 \text{ a.s.  for all } i\in\{1,...,k\},
        \end{equation}
        and such that,
        \begin{equation}\label{eq:CompThm2}
            g(t,\hat{Y}_t,\hat{Z}_t,K_t^1,\ldots,K_t^p) - g(t,\hat{Y}_t,\hat{Z}_t,\hat{K}_t^1,\ldots,\hat{K}_t^p) \ge \sum_{i=1}^p\gamma_t^i(K_t^i-\hat{K}_t^i)\lambda_t^i
        \end{equation}
        for $t\in[0,T]$, $dP\otimes dt$-a.e. Suppose that $\eta\ge\hat{\eta}$ a.s.  that $\bar{D}\coloneqq D - \hat{D}$ is non-decreasing, and that
        \begin{equation}\label{eq:CompThm3}
            g(t,\hat{Y}_t,\hat{Z}_t,\hat{K}_t^1,\ldots,\hat{K}_t^p)\ge \hat{g}(t,\hat{Y}_t,\hat{Z}_t,\hat{K}_t^1,\ldots,\hat{K}_t^p)
        \end{equation}
        for $t\in[0,T]$, $dP\otimes dt$-a.e.
        We then have $Y_t\ge\hat{Y}_t$ for all $t\in[0,T]$ a.s.
        \item \textbf{Strict Comparison:} Assume moreover that   $\gamma_{\tau_i}^i>-1$ a.s. for each $i\in\{1,\ldots,p\}$ and  that there exists $t_0\in [0,T]$ such that  $Y_{t_0}=\hat{Y}_{t_0}$ a.s.  Then, $\eta=\hat{\eta}$ a.s. and the inequality in \eqref{eq:CompThm3} is an equality on $[t_0,T]$. Furthermore, $\bar{D}\coloneqq D - \hat{D}$ is constant on $[t_0,T]$ and $Y =\hat{Y}$ on $[t_0,T]$.
    \end{enumerate}
\end{theorem}
\begin{remark}\label{remark:CompThm997}
Due to the assumption $\tau_1<\tau_2<...<\tau_p$, the condition from Eq. \eqref{eq:CompThm1}, namely,  for each  $k\in\{1,...,p\}$,   on  $A_k, 1+\gamma_{\tau_i}^i\ge0$  a.s.  for all  $i\in\{1,...,k\}$, is equivalent, in our framework, to the condition:
\begin{equation}\label{old_99876}
\text{ for all } t\in[0,T], \prod_{i=1}^p(1+\gamma_{\tau_i}^i\mathbbm{1}_{\{t<\tau_i\le s\}})\ge0, \text{ for all } s\in[t,T], \text{ a.s.},
\end{equation}
  which ensures the non-negativity of the adjoint process $(\Gamma_{t,\cdot})$ in the proof of the comparison theorem.   To show the equivalence between the two conditions, we proceed as follows: Let the condition from Eq.\eqref{old_99876} hold. Let us take t=0 in this condition, and let $k\in\{1,...,p\}$. Let us place ourselves on the set $A_k$: taking successively $s=\tau_1(\omega)$, $s=\tau_2(\omega)$,..., $s= \tau_k(\omega)$ in Eq.\eqref{old_99876} (and using that  $\tau_1<\tau_2<...<\tau_p$), we get,  $1+\gamma^1_{\tau_1(\omega)}(\omega)\geq 0$, ..., $1+\gamma^k_{\tau_k(\omega)}(\omega)\geq 0$ on $A_k$. Conversely, let the condition from Eq. \eqref{eq:CompThm1} hold true. For $t\in[0,T]$,   for $s\in[t,T]$, we have $\prod_{i=1}^p(1+\gamma_{\tau_i}^i\mathbbm{1}_{\{t<\tau_i\le s\}})=\sum_{k=0}^p \mathbbm 1_{A_k} (\prod_{i=1}^p(1+\gamma_{\tau_i}^i\mathbbm{1}_{\{t<\tau_i\le s\}})).$ Let $k\in\{1,...,p-1\}. $ We consider the set  $A_k\cap \{t<\tau_i\leq s\}= \{\tau_k\leq T,\tau_{k+1}> T\}\cap \{t<\tau_i\leq s\}$: for each $i$ such that  $1 \leq i\leq k$, it holds $\gamma_{\tau_i}^i\geq -1$ on $A_k\cap \{t<\tau_i\leq s\}$ (by condition  \eqref{eq:CompThm1}); for $i\geq k+1$, $A_k\cap \{t<\tau_i\leq s\}=\varnothing$ (as $\tau_{k+1}>T$ on $A_k$ and as the $\tau_i$'s are strictly ordered). On $A_p$, $\gamma_{\tau_i}^i\geq -1,$ for each $i\in\{1,...,p\}.$  Finally, we note that $A_0\cap \{t<\tau_i\leq s\}=\varnothing$, for each $i\in\{1,...,p\}$ (as $\tau_1>T$ on $A_0$).    We  conclude that, for $t\in[0,T]$,   for  $s\in[t,T]$,  $\prod_{i=1}^p(1+\gamma_{\tau_i}^i\mathbbm{1}_{\{t<\tau_i\le s\}})\geq 0.$  
  \end{remark}
\begin{remark}\label{remark:CompThm}
	Assume that  $\gamma_{\tau_i}^i\ge-1$ a.s. for each $i\in\{1,\ldots,p\}$. This implies that  the condition:    for each $k\in\{1,...,p\}$,   on  $A_k$, $1+\gamma_{\tau_i}^i\ge0$  a.s.  for all  $i\in\{1,...,k\}$,  from Eq. \eqref{eq:CompThm1} is satisfied.  If, moreover, $\gamma_{\tau_i}^i>-1$ a.s. for each $i\in\{1,\ldots,p\}$, then  the condition 
	from the strict comparison (ii) is also  satisfied. 
\end{remark}

\begin{proof}
	Setting $\bar{Y}_s \coloneqq Y_s - \hat{Y}_s$, $\bar{Z}_s\coloneqq Z_s - \hat{Z}_s$ and $\bar{K}_s^i\coloneqq K_s^i - \hat{K}_s^i$, for each $i\in\{1,\ldots,p\}$, we have,
    \begin{equation*}
        -d\bar{Y}_s = h_sds + d\bar{D}_s - \bar{Z}_sdW_s - \sum_{i=1}^p\bar{K}_s^idM_s^i,\quad\bar{Y}_T = \eta - \hat{\eta},
    \end{equation*}
    where,
\begin{equation}\label{eq:hequation}
h_s\coloneqq g(s,Y_{s-},Z_s,K_s^1,\ldots,K_s^p)-\hat{g}(s,\hat{Y}_{s-},\hat{Z}_s,\hat{K}_s^1,\ldots,\hat{K}_s^p).
\end{equation}
    We set,
    \begin{align}\label{eq:deltaequation}
	\begin{aligned}
	        \delta_s&\coloneqq\frac{g(s,Y_{s-},Z_s,K_s^1,\ldots,K_s^p) - g(s,\hat{Y}_{s-},Z_s,K_s^1,\ldots,K_s^p)}{\bar{Y}_{s-}}\mathbbm{1}_{\{\bar{Y}_{s-}\ne0\}},\\
	        \beta_s&\coloneqq\frac{g(s,\hat{Y}_{s-},Z_s,K_s^1,\ldots,K_s^p)-g(s,\hat{Y}_{s-},\hat{Z}_s,K_s^1,\ldots, K_s^p)}{\bar{Z}_s}\mathbbm{1}_{\{\bar{Z}_{s}\ne0\}}.
	\end{aligned}
    \end{align}
By definition both $\delta$ and $\beta$ are predictable. Furthermore, since $g$ is a $\lambda^{(p)}$-admissible driver,  it satisfies,
\begin{multline*}
    |g(\omega,t,y,z,k^1,\ldots,k^p) - g(\omega, t, \hat{y}, \hat{z}, \hat{k}^1,\ldots,\hat{k}^p)|
    \\\le C\left(|y-\hat{y}| + |z-\hat{z}| + \sum_{i=1}^p\sqrt{\lambda_t^i(\omega)}|k^i-\hat{k}^i|\right);
\end{multline*}
hence, the processes $\delta$ and $\beta$ are bounded. With the above notation,
\begin{equation*}
    h_s = \delta_s\bar{Y}_{s-} + \beta_s\bar{Z}_s + g(s,\hat{Y}_{s-},\hat{Z}_s,K_s^1,\ldots,K_s^p) - g(s,\hat{Y}_{s-},\hat{Z}_s,\hat{K}_s^1,\ldots,\hat{K}_s^p) + \varphi_s,
\end{equation*}
where,
\begin{equation}\label{eq:varphi1}
	\varphi_s \coloneqq g(s,\hat{Y}_{s-},\hat{Z}_s,\hat{K}_s^1,\ldots,\hat{K}_s^p) - \hat{g}(s,\hat{Y}_{s-},\hat{Z}_s,\hat{K}_s^1,\ldots,\hat{K}_s^p).
\end{equation}

Due to  assumption \eqref{eq:CompThm2} and due to the fact that $Y_t 
= Y_{t-}$ $dP\otimes dt$-a.e.\footnote{This is true as $Y_t(\omega)$ has at most a countable number of jumps.}, we have
\begin{equation}\label{eq:proofCompThm1}
    h_s\ge\delta_s\bar{Y}_s + \beta_s\bar{Z}_s + \sum_{i=1}^p\gamma_s^i\bar{K}_s^i\lambda_s^i + \varphi_s,\quad dP\otimes ds\text{-a.e.}
\end{equation}

We fix $t\in[0,T]$. Let $\Gamma_{t,.}$ be the adjoint process, defined by,
\begin{equation*}
    d\Gamma_{t,s} = \Gamma_{t,s-}\left(\delta_sds+\beta_sdW_s+\sum_{i=1}^p\gamma_s^idM_s^i\right),\quad\Gamma_{t,t}=1.
\end{equation*}
As $\delta$, $\beta$ and $\gamma^i\sqrt{\lambda^i}$ for each $i\in\{1,\ldots,p\}$ are bounded, we have that $\Gamma_{t,\cdot}\in\mathcal{S}^2$ by Remark \ref{remark:thmRR_Dpred}. Due to the condition from Eq. \eqref{eq:CompThm1}, to Equation \eqref{eq:remarkAdjoint1} and to Remark \ref{remark:CompThm997}, we have $\Gamma_{t,s}\ge0$  for all $t\le s\le T$ a.s.  

\textit{Step 1. } We  consider first the case where  $D$ and $\hat{D}$ are (predictable)   processes in $\mathcal{A}_{p,T}^2$, and prove the comparison and the strict comparison results in this case.\\ 
By applying It\^o's product rule to $(\bar{Y}_s\Gamma_{t,s})$, we get, $-d(\bar{Y}_s\Gamma_{t,s})=-\bar{Y}_{s-}d\Gamma_{t,s} -\Gamma_{t,s-}d\bar{Y}_s - d[\bar{Y},\Gamma_{t,\cdot}]_s$. We have, 
\begin{align}\label{eq:proofCompThm2}
    \begin{aligned}
        d[\bar{Y},\Gamma_{t,\cdot}]_s&= d\left[\int_t^\cdot\bar{Z}_rdW_r,\int_t^\cdot\Gamma_{t,r-}\beta_rdW_r\right]_s + d\left(\sum_{i=1}^p\sum_{j=1}^p\left[\int_t^\cdot\bar{K}_r^idM_r^i,\int_t^\cdot\Gamma_{t,r-}\gamma_r^jdM_r^j\right]_s\right)\\
	&= \Gamma_{t,s-}\beta_s\bar{Z}_sds + \left(\int_t^s\sum_{i=1}^p\sum_{j=1}^p\bar{K}_r^i\Gamma_{t,r-}\gamma_r^jd[M^i,M^j]_r\right)\\
        &=\Gamma_{t,s-}\beta_s\bar{Z}_sds + d\left(\int_t^s\Gamma_{t,r-}\sum_{i=1}^p\bar{K}_r^i\gamma_r^idN_r^i\right)\\
        &= \Gamma_{t,s-}\beta_s\bar{Z}_sds + \Gamma_{t,s-}\sum_{i=1}^p\bar{K}_s^i\gamma_s^i\lambda_s^ids + \Gamma_{t,s-}\sum_{i=1}^p\bar{K}_s^i\gamma_s^idM_s^i,
    \end{aligned}
\end{align}
where we have used that $d[M^i,M^j]_s = 0$, for $i\ne j$, since $P(\tau_i = \tau_j)=0$, for $i\ne j$, and for the case $i=j$,  $d[M^i]_s =  dN_s^i$. This yields,
\begin{align}\label{eq:proofCompThm3}
    \begin{aligned}
        -d(\bar{Y}_s\Gamma_{t,s}) &= -\Gamma_{t,s-}\left(\bar{Y}_{s-}\delta_sds + \bar{Y}_{s-}\beta_sdW_s + \bar{Y}_{s-}\sum_{i=1}^p\gamma_s^idM_s^i\right) \\&\quad+ \Gamma_{t,s-}\left(h_sds + d\bar{D}_s - \bar{Z}_sdW_s - \sum_{i=1}^p\bar{K}_s^idM_s^i\right)\\&\quad-\Gamma_{t,s-}\beta_s\bar{Z}_sds -\Gamma_{t,s-}\sum_{i=1}^p\bar{K}_s^i\gamma_s^i\lambda_s^ids -\Gamma_{t,s-}\sum_{i=1}^p\bar{K}_s^i\gamma_s^idM_s^i\\
        &=\Gamma_{t,s-}\left(h_s-\delta_s\bar{Y}_{s-} - \beta_s\bar{Z}_s-\sum_{i=1}^p\bar{K}_s^i\gamma_s^i\lambda_s^i\right)ds + \Gamma_{t,s-}d\bar{D}_s\\&\quad-\left(\Gamma_{t,s-}(\bar{Y}_{s-}\beta_s+\bar{Z}_s)dW_s + \Gamma_{t,s-}\left(\sum_{i=1}^p(\bar{K}_s^i(1+\gamma_s^i)+\bar{Y}_{s-}\gamma_s^i)dM_s^i\right)\right)\\
        &=\Gamma_{t,s-}\left(h_s-\delta_s\bar{Y}_{s-} - \beta_s\bar{Z}_s-\sum_{i=1}^p\bar{K}_s^i\gamma_s^i\lambda_s^i\right)ds + \Gamma_{t,s-}d\bar{D}_s-dm_s,
    \end{aligned}
\end{align}

where the process $(m_s)_{s\in[0,T]}$ is defined by 
\begin{equation}\label{eq:proofCompThm_m}
	dm_s = \Gamma_{t,s-}(\bar{Y}_s\beta_s+\bar{Z}_s)dW_s + \Gamma_{t,s-}(\sum_{i=1}^p(\bar{K}_s^i(1+\gamma_s^i)+\bar{Y}_{s-}\gamma_s^i)dM_s^i).
\end{equation}
The process $(m_s)$ is a martingale, since $\Gamma_{t,\cdot}\in\mathcal{S}^2$, $\bar{Y}\in\mathcal{S}^2$, $\bar{Z}\in\mathcal{H}^2$, $\bar{K}^i\in\mathcal{H}_{\lambda^i}^2$ for each $i\in\{1,\ldots,p\}$, and since $\beta$ and $\gamma^i\sqrt{\lambda^i}$, for each $i\in\{1,\ldots,p\}$, are bounded.
 Using Equations \eqref{eq:proofCompThm1} and \eqref{eq:proofCompThm3}, and the fact that $\Gamma$ is non-negative, we get,
\begin{equation}\label{eq:proofCompThm4}
    -d(\bar{Y}_s\Gamma_{t,s}) \ge \Gamma_{t,s-}\varphi_sds + \Gamma_{t,s-}d\bar{D}_s - dm_s.
\end{equation}

Integrating \eqref{eq:proofCompThm4} between $t$ and $T$, and taking the conditional expectation, results in,
\begin{equation}\label{eq:proofCompThm5}
    \bar{Y}_t \ge \mathbb{E}\left[\Gamma_{t,T}(\eta-\hat{\eta}) + \int_t^T\Gamma_{t,s-}\varphi_sds + \int_t^T\Gamma_{t,s-}d\bar{D}_s\middle|\mathcal{G}_t\right],\quad0\le t\le T\text{ a.s.}
\end{equation}

From \eqref{eq:CompThm3} we have that $\varphi_t\ge0$ $dP\otimes dt$-a.e. Furthermore, since $\eta-\hat{\eta}\ge0$, since $\bar{D}$ is non-decreasing and the adjoint process $(\Gamma_{t,s})_{s\in [t,T]}$ is non-negative, we have that all terms inside the conditional expectation are non-negative; hence,  $\bar{Y}_t = Y_t - \hat{Y}_t\ge0$ a.s. Since this holds for all $t\in[0,T]$ and since both $Y$ and $\hat Y$ are rcll,  the \emph{comparison} result (i) for $D,\hat{D}\in\mathcal{A}_{p,T}^2$ is proven. \\
Let us prove (ii)  for $D,\hat{D}\in\mathcal{A}_{p,T}^2$. 
Assume  that there exists  $t_0\in[0,T]$ such that $Y_{t_0}=\hat{Y}_{t_0}$ a.s. and such  that $\prod_{i=1}^p(1+\gamma_{\tau_i}^i\mathbbm{1}_{t_0<\tau_i\le s})>0$ for all $s\in[t_0,T]$ a.s. (cf. Eq. \eqref{eq:remarkAdjoint1}).  This implies that  $\Gamma_{t_0,s}>0$ for all $s\in[t_0,T]$ a.s. On the other hand, Equation  \eqref{eq:proofCompThm5} (for $t=t_0$) leads to:
\begin{equation}\label{eq:proofCompThm6}
    0 =\bar{Y}_{t_0}\ge \mathbb{E}\left[\Gamma_{t_0,T}(\eta-\hat{\eta}) + \int_{t_0}^T\Gamma_{t_0,s-}\varphi_sds+\int_{t_0}^T\Gamma_{t_0,s-}d\bar{D}_s\middle|\mathcal{G}_{t_0}\right].
\end{equation}
This, together with the non-negativity of the terms inside the conditional expectation and the positivity of $(\Gamma_{t,s})$, implies that $\eta=\hat{\eta}$ a.s. and $\varphi_t=0$, for all $t\in[t_0,T]$ $dP\otimes dt$-a.e. Let us now set $\Tilde{D}_t\coloneqq\int_{t_0}^t\Gamma_{t_0,s-}d\bar{D}_s$ for each $t\in[t_0,T]$. We have $\Tilde{D}_T\ge0$ a.s. as $\Gamma_{t_0,s}>0$ and as $\bar{D}$ is non-decreasing (by assumption). Using this and \eqref{eq:proofCompThm6}, we get $0 = \mathbb{E}[\Tilde{D}_T|\mathcal{G}_{t_0}]$ a.s., hence $\Tilde{D}_T = 0$ a.s. Since $\Gamma_{t_0,s}>0$ for all $T \ge s\ge t_0$ a.s., we have $\bar{D}_T - \bar{D}_{t_0} = \int_{t_0}^T(\Gamma_{t_0,s-})^{-1}d\Tilde{D}_s$, which implies that $\bar{D}_{t_0} = \bar{D}_T$ a.s. Hence, the \emph{strict comparison} result (ii) is proven fro $D,\hat{D}\in\mathcal{A}_{p,T}^2$.\\
\textit{Step 2.} We now consider the case where the optional processes $D,\hat{D}$ are not necessarily predictable; more precisely, $D,\hat{D}\in\mathcal{A}_T^2$.  By Proposition \ref{prop:OptionalpJumps} (applied to $D$ and to $\hat{D}$) there exist $D'\in\mathcal{A}_{p,T}^2$, $\hat{D}'\in\mathcal{A}_{p,T}^2$ and $\theta^i,\hat{\theta}^i\in\mathcal{H}_{\lambda^i,T}^2$ (where $i\in\{1,\ldots,p\}$), such that $D$ and $\hat{D}$ can be uniquely written as follows,
\begin{align}\label{eq:proofCompThm7}
    \begin{aligned}
        D_t &= D_t' + \int_0^t\sum_{i=1}^p\theta_s^idN_s^i,\quad\text{a.s. and }
        \hat{D}_t = \hat{D}_t' + \int_0^t\sum_{i=1}^p\hat{\theta}_s^idN_s^i,\quad\text{a.s.}
    \end{aligned}
\end{align}
Since $\bar{D}\coloneqq D - \hat{D}$ is non-decreasing, and since $\tau_1<\cdots<\tau_p$, by Lemma \ref{lemma:nonDecreasingDecompositions} we get that $\bar{D}' \coloneqq D' - \hat{D}'$ is non-decreasing and for each $i\in\{1,\ldots,p\}$ $\theta_{\tau_i}^i\ge\hat{\theta}_{\tau_i}^i$ a.s. on $\{\tau_i\le T\}$. By applying It\^o's product rule to $(\bar{Y}_s\Gamma_{t,s})$, we get $-d(\bar{Y}_s\Gamma_{t,s}) = -\bar{Y}_{s-}d\Gamma_{t,s}-\Gamma_{t,s-}d\bar{Y}_s - d[\bar{Y},\Gamma_{t,\cdot}]_s$. Here, $d[\bar{Y},\Gamma_{t,\cdot}]_s$ is equal to the right-hand side of \eqref{eq:proofCompThm2} plus the additional term $d[\bar{D}, \Gamma_{t,\cdot}]_s$. The term $d[\bar{D}, \Gamma_{t,\cdot}]_s$ can be expressed as
\begin{align}\label{eq:proofCompThm8}
    \begin{aligned}
        d[\bar{D},\Gamma_{t,\cdot}]_s &= d[\bar{D}',\Gamma_{t,\cdot}]_s + d\left(\sum_{i=1}^p\sum_{j=1}^p\left[\int_t^\cdot(\theta_r^i-\hat{\theta}_r^i)dN_r^i,\int_t^\cdot\Gamma_{t,r-}\gamma_r^jdN_r^j\right]_s\right)\\
	&=d\left(\int_t^s\Gamma_{t,r-}\sum_{i=1}^p\sum_{j=1}^p(\theta_r^i-\hat{\theta}_r^i)\gamma_r^jd[N^i,N^j]_r\right)\\
        &= d\left(\int_t^s\Gamma_{t,r-}\sum_{i=1}^p(\theta_r^i-\hat{\theta}_r^i)\gamma_r^idN_r^i\right)= \Gamma_{t,s-}\sum_{i=1}^p(\theta_s^i-\hat{\theta}_s^i)\gamma_s^idN_s^i,
    \end{aligned}
\end{align}
where we have used that $d[N^i,N^j]_s = 0$, for $i\ne j$, since $P(\tau_i = \tau_j)=0$, , for $i\ne j$, and $d[N^i]_s = dN_s^i$ (when $i=j$). Hence, we have,
\begin{multline}\label{eq:proofCompThm9}
    -d(\bar{Y}_s\Gamma_{t,s}) = \Gamma_{t,s-}\left(h_s - \delta_s\bar{Y}_s - \beta_s\bar{Z}_s-\sum_{i=1}^p\bar{K}_s^i\gamma_s^i\lambda_s^i\right)ds +\Gamma_{t,s-}d\bar{D}_s-dm_s\\
    +\Gamma_{t,s-}\sum_{i=1}^p(\theta_s^i-\hat{\theta}_s^i)\gamma_s^idN_s^i,
\end{multline}
where $(m_t)$ is the same martingale as the one from Eq. \eqref{eq:proofCompThm_m}, $(h_t)$ is the  process from Eq.  \eqref{eq:hequation}, and $(\delta_t)$ and $(\beta_t)$ are the  processes from Eq. \eqref{eq:deltaequation}. Using inequality \eqref{eq:proofCompThm1} and the fact that $d\bar{D}_t = d\bar{D}_t'+\sum_{i=1}^p(\theta_t^i-\hat{\theta}_t^i)dN_t^i$, integrating between $t$ and $T$ (where $t\in[0,T]$), and taking the conditional expectation, we get,
\begin{equation}\label{eq:proofCompThm10}
    \bar{Y}_t \ge \mathbb{E}\left[\Gamma_{t,T}(\eta-\hat{\eta}) + \int_t^T\Gamma_{t,s-}\left(d\bar{D}_s'+\sum_{i=1}^p(\theta_s^i - \hat{\theta}_s^i)(1+\gamma_s^i)dN_s^i+\varphi_sds\right)\middle|\mathcal{G}_t\right]\text{,  a.s.}
\end{equation}
where the process $(\varphi_t)$ is the same as the one from Eq.  \eqref{eq:varphi1}.  Let us note that, for $i\in\{1,\ldots,p\}$, 
\begin{align}\label{new_eq_9087}
	\int_t^T\Gamma_{t,s-}(\theta_s^i-\hat{\theta}_s^i)(1+\gamma_s^i)dN_s^i = \Gamma_{t,\tau_i-}(\theta_{\tau_i}^i-\hat{\theta}_{\tau_i}^i)(1+\gamma_{\tau_i}^i)\mathbbm{1}_{\{T\ge\tau_i\ge t\}}.
\end{align}

Let $i$ be fixed. We now check that this term is non-negative on each $A_k$, where $k\in\{1,...,p\}$, and where the $A_k$'s are the ones appearing in  assumption \eqref{eq:CompThm1}.  	As noted above,  $\theta_{\tau_i}^i\ge\hat{\theta}_{\tau_i}^i$ a.s. on $\{\tau_i\le T\}$. Furthermore, the adjoint process  $(\Gamma_{t,s})_{s\in[0,T]}$ is non-negative.  Moreover, by definition of $A_k$, we have   $\mathbbm{1}_{\{T\ge\tau_i\ge t\}}\mathbbm{1}_{A_k}=0,$ for $ 0\leq k\leq i-1$, and,   by assumption  \eqref{eq:CompThm1}, we have $(1+\gamma^i_{\tau_i})\mathbbm{1}_{\{T\ge\tau_i\ge t\}}\mathbbm{1}_{A_k}\geq 0$ (for each $k\in\{i,...,p\}$). Hence, the term in Eq. \eqref{new_eq_9087} is non-negative. 
From the assumption \eqref{eq:CompThm3} and from  \eqref{eq:varphi1} we have that $\varphi_t\ge0$ $dP\otimes dt$-a.e. Furthermore, since $\eta-\hat{\eta}\ge0$, since $(\bar{D}_t')$ is non-decreasing and the adjoint process $(\Gamma_{t,s})_{s\in[0,T]}$ is non-negative, we have that all the terms inside the conditional expectation are non-negative; hence, $\bar{Y}_t = Y_t-\hat{Y}_t\ge0$ a.s. Since this holds for all $t\in[0,T]$, and since $(Y_t)$ and $(\hat Y_t)$ are rcll,  the \emph{comparison result} (i) for $D,\hat{D}\in\mathcal{A}_T^2$ is proven. 

Assume now that there exists $t_0\in[0,T]$ such that $Y_{t_0} = \hat{Y}_{t_0}$ a.s. and that for each $i\in\{1,\ldots,p\}$ $\gamma^i_{\tau_i}>-1$ a.s. Thus,  $\Gamma_{t,s}>0$ for all $ s\in[t,T]$ a.s. For $t=t_0$, Eq. \eqref{eq:proofCompThm10} leads to,
\begin{equation}\label{eq:proofCompThm100}
    0 = \bar{Y}_{t_0} \ge \mathbb{E}\left[\Gamma_{t_0,T}(\eta-\hat{\eta}) + \int_{t_0}^T\Gamma_{t,s-}\left(d\bar{D}_s'+\sum_{i=1}^p(\theta_s^i - \hat{\theta}_s^i)(1+\gamma_s^i)dN_s^i+\varphi_sds\right)\middle|\mathcal{G}_{t_0}\right].
\end{equation}
This, together with the non-negativity of the terms inside the conditional expectation and the positivity of $(\Gamma_{t,s})$, implies that $\eta=\hat{\eta}$ a.s., $\varphi_t=0$ for all $t\in[t_0,T]$ $dP\otimes dt$-a.e., and, for each $i\in\{1,\ldots,p\}$, $\theta_{\tau_i}^i=\hat{\theta}_{\tau_i}^i$ on $\{t_0<\tau_i\le T\}$ a.s. We  set $\Tilde{D}_t'\coloneqq\int_{t_0}^t\Gamma_{t_0,s-}d\bar{D}_s'$ for each $t\in[t_0,T]$. We have $\Tilde{D}_T'\ge0$ a.s. as $\Gamma_{t_0,s}>0$ and as $\bar{D}'$ is non-decreasing. Using this and \eqref{eq:proofCompThm100}, we get $\mathbb{E}[\Tilde{D}_T'|\mathcal{G}_{t_0}]=0$ a.s.; hence, $\Tilde{D}_T'=0$ a.s. Since $\Gamma_{t_0,s}>0$ for all $T\ge s\ge t_0$ a.s. we have $\bar{D}_T' - \bar{D}_{t_0}' = \int_{t_0}^T(\Gamma_{t_0,s-})^{-1}d\Tilde{D}_s'$, which implies that $\bar{D}_{t_0}'=\bar{D}_T'$ a.s. The \emph{strict comparison} result (ii) for $D,\hat{D}\in\mathcal{A}_T^2$ is thus proven.

\end{proof}

We now provide an example where the conclusion of the comparison (and strict comparison) result from Theorem \ref{thm:CompThm} does not necessarily hold, if the assumptions of the theorem are not satisfied.

\begin{example}\label{example:exampleCounterExample2}
Assume  that for each $i\in\{1,\ldots,p\}$ the process $\lambda^i$ is bounded. Let $g$ be a $\lambda^{(p)}$-linear driver (cf. Definition \ref{eq:defLambdaLinear}) of the form,
    \begin{equation}\label{eq:exampleComp1}
        g(\omega,t,y,z,k^1,\ldots,k^p) = \alpha_t(\omega)y + \beta_t(\omega)z + \sum_{i=1}^p\gamma^ik^i\lambda_t^i(\omega),
    \end{equation}
    where each $\gamma^i$ is a real constant. The dynamics of the  adjoint process $\Gamma_{0,\cdot}$ are  (cf. \eqref{eq:ThmDpred2}),
    \begin{equation}\label{eq:exampleComp2}
        d\Gamma_{0,s} = \Gamma_{0,s-}\left(\alpha_sds + \beta_sdW_s+\sum_{i=1}^p\gamma^idM_s^i\right),\quad\Gamma_{0,0}=1.
    \end{equation}
    By Remark \ref{remark:thmRR_Dpred},  $\Gamma_{0,T}$ satisfies,
    \begin{equation}\label{eq:exampleComp3}
        \Gamma_{0,T} = H_T\exp\left(-\int_0^T\sum_{i=1}^p\gamma^i\lambda_s^ids\right)\prod_{i=1}^p(1+\gamma^i\mathbbm{1}_{\{0<\tau_i\le T\}}),
    \end{equation}
    where $H$ has the dynamics $dH_t = H_t(\alpha_tdt + \beta_tdW_t)$ with $H_0=1$.\\
We specify $p=2$. We define the terminal condition as,
    \begin{equation}\label{eq:exampleTerminal2}
        \eta^{(1)} \coloneqq \mathbbm{1}_{\{\tau_1\leq T,\tau_2>T\}}.
    \end{equation}
    
  Let $(Y^{(1)})$ be the first component of the solution of the BSDE associated with  driver $g$, terminal time $T$ and terminal condition $\eta^{(1)}$. By the explicit formula  from Theorem \ref{theorem:RRGLB}, we get,

  $$Y_0^{(1)} = (1+\gamma^1)\mathbb{E}\left[H_Te^{-\sum_{j=1}^2\gamma^j\int_0^T\lambda_s^jds}\mathbbm{1}_{\{\tau_1\leq T,\tau_2>T\}}\right].$$
  
  Under the assumption $P(\tau_1\leq T,\tau_2>T)>0$, if $1+\gamma^1<0$, then  $Y_0^{(1)}<0$. However,  $\eta^{(1)}\ge0$ a.s. Hence, the comparison result does not hold.\\
We now define the terminal condition as,
    \begin{equation}\label{eq:exampleTerminal2}
        \eta^{(2)} \coloneqq \mathbbm{1}_{\{\tau_2\leq T\}}.
    \end{equation}
    
    Let $(Y^{(2)})$ be the solution of the BSDE associated with  driver $g$, terminal time $T$ and terminal condition $\eta^{(2)}$. 

    By the explicit formula  from Theorem \ref{theorem:RRGLB},
    \begin{equation}\label{eq:exampleRepResult2}
        Y_0 = (1+\gamma^1)(1+\gamma^2)\mathbb{E}\left[H_Te^{-\sum_{j=1}^2\gamma^j\int_0^T\lambda_s^jds}\mathbbm{1}_{\{T\ge\tau_2\}}\right],
    \end{equation}
where we have used that $\tau_1<\tau_2$.
    Under the assumption $P(\tau_2\leq T)>0$, if $(1+\gamma^1)(1+\gamma^2)<0$, then \eqref{eq:exampleRepResult2} leads to $Y_0<0$. However,  $\eta^{(2)}\ge0$ a.s. Hence, the comparison result does not hold.\\
    The reader can generalize this reasoning to the case where $p>2$, by using terminal conditions based on the sets  from Eq.\eqref{sets}.

    If either (or both)  $\gamma^1$ and $\gamma^2$ are equal to $-1$, then \eqref{eq:exampleRepResult2} gives $Y_0=0$. Under the assumption that $P(\tau_2\leq T)>0$, we have $P(\eta^{(2)}>0)>0$, while $Y_0=0$, hence the strict comparison result does not hold.
\end{example}

\section{Pricing of European Options in Markets with Multiple Defaults}\label{sec:Pricing}

\subsection{ Pricing in a Linear Financial Market with two Defaultable Risky Assets}\label{subsec:LinearPricing1}

We consider a market model where the primary assets are a risk-free savings account with price process $B$, a default-free asset with price process $S^0$, and two assets with price processes  $S^1$ and $S^2$,  which are subject to default or to some other credit event at times $\tau^1$ and $\tau^2$, respectively.

More precisely, we place ourselves in the probabilistic setting of Section \ref{sec:PPS}, where we set $p=2$. The times $\tau_1$ and $\tau_2$ model here the times of default (or the times of some other extraneous credit events, provided they are ordered) of the risky assets $S^1$ and $S^2$, respectively. As before $M_t^1 = N_t^1 - \int_0^t\lambda_s^1ds$ and $M_t^2 = N_t^2 - \int_0^t\lambda_s^2ds$. \\
We consider the following dynamics for the asset prices,
\begin{align}\label{eq:3.1}
    \begin{aligned}
    dB_t &= B_tr_tdt,\quad B_0=1;\\
    dS_t^0 &= S_t^0[\mu_t^0dt + \sigma_t^0dW_t],\quad S_0^0>0;\\
    dS_t^1 &= S_{t-}^1[\mu_t^1dt + \sigma_t^1dW_t + \beta_t^1dM_t^1],\quad S_0^1>0;\\
    dS_t^2 &= S_{t-}^2[\mu_t^2dt + \sigma_t^2dW_t + \beta_t^2dM_t^2],\quad S_0^2>0.
    \end{aligned}
\end{align}The process $r$, and the processes $\mu^i$ and $\sigma^i$ (for $i\in\{0,1,2\}$) are predictable, such that, $\sigma^i>0$ for $i\in\{0,1,2\}$, and $r,\mu^i,\sigma^i$ and $(\sigma^i)^{-1}$ are bounded (for $i\in\{0,1,2\}$). We note that there is no requirement for the intensity process $\lambda^i$ to be bounded. We assume that $\beta_t^1\neq 0$,  $\beta_t^2\neq 0$, and $\mu_t^0\neq r_t$. We assume moreover that $\beta_t^i\ge-1$ for $i\in\{1,2\}$.

\begin{remark}\label{remark:PriceProcess}

    By Remark \ref{remark:2.12}, the  explicit formula for $S^i$, where $i\in\{1,2\}$, is: for $t\in[0,T]$,
    \begin{equation}\label{eq:remarkPriceProcess}
        S_t^i = \exp\left(\int_0^t\left(\mu_s^i- \frac{1}{2}(\sigma_s^i)^2-\beta_s^i\lambda_s^i\right)ds+\int_0^t\sigma_s^idW_s \right)(1+\beta_{\tau_i}^i\mathbbm{1}_{\{t\ge\tau_i\}}),\text{  a.s.}
    \end{equation}

    If $\beta^i_{\tau_i}=-1$,  then  the $i$-th asset's price jumps to zero at $\tau_i$.
    
\end{remark} 

We consider an investor who at time $0$ invests an amount $x\in\mathbb{R}$ in the market. For $i\in\{0,1,2\}$ we use $\phi_t^i$ to denote the amount of money in asset $S_t^i$ at time $t\in[0,T]$.\\
 If $\beta_{\tau_i}^i=-1$ a.s., then, on the set $\{T\ge t\ge\tau_i\}$, $S_t^i=0$ and  the investor will no longer invest in this asset; thus, $\phi_t^i=0$ on the set $\{T\ge t>\tau_i\}$. \\
In the case where $p=1$ and $\beta^1=-1$, this model has been considered in  \cite{bielecki2} and \cite{dumitrescu2018bsdes}}. 

The process $\phi = (\phi^0,\phi^1,\phi^2)\in(\mathcal{H}_T^2,\mathcal{H}_{\lambda^1,T}^2,\mathcal{H}_{\lambda^2,T})$ is called the risky-asset strategy (or the strategy). Let now $(C_t)_{t\in[0,T]}$ be a finite variational optional process in $\mathcal{A}_T^2$ which represents the cumulative  cash `withdrawals' from the portfolio. The value of the portfolio at time $t$ associated with the initial value $x$, trading strategy $\phi$ and `withdrawal' process $C$ is denoted by $V_t^{x,\phi,C}$. If $\phi$ is the strategy in the risky assets $S^0, S^1, S^2$, then the amount invested in the risk-free bank account is: $V_t^{x,\phi,C}-\sum_{i=0}^2\phi^i_t$.


The self-financing condition for the wealth process $V^{x,\phi,C}=V$ leads to the following dynamics:

\begin{align}\label{eq:3_SFC}
    \begin{aligned}
        dV_t &= \left(\frac{V_t-\sum_{i=0}^2\phi_t^i}{B_t^0}\right)dB_t^0 + \sum_{i=0}^2\frac{\phi_t^idS_t^i}{S_{t-}^i} - dC_t\\
        &=(V_t-\phi_t^0-\phi_t^1-\phi_t^2)r_tdt + \phi_t^0(\mu_t^0dt+\sigma_t^0dW_t) \\&\quad\quad+ \phi_t^1(\mu_t^1dt+\sigma_t^1dW_t+\beta_t^1dM_t) + \phi_t^2(\mu_t^2dt+\sigma_t^2dW_t+\beta_t^2dM_t^2) - dC_t\\
        &=(V_tr_t+\phi_t^0(\mu_t^0-r_t) + \phi_t^1(\mu_t^1-r_t)+\phi_t^2(\mu_t^2-r_t))dt\\
        &\quad\quad+(\phi_t^0\sigma_t^0+\phi_t^1\sigma_t^1+\phi_t^2\sigma_t^2)dW_t + \phi_t^1\beta_t^1dM_t^1+\phi_t^2\beta_t^2dM_t^2 - dC_t\\
        &=(V_tr_t + \phi_t'\sigma_t\Theta_t^0+\phi_t^1\Theta_t^1\beta_t^1\lambda_t^1 + \phi_t^2\Theta_t^2\beta_t^2\lambda_t^2)dt - dC_t + \phi_t'\sigma_tdW_t \\
	 &\quad\quad+ \phi_t^1\beta_t^1dM_t^1 + \phi_t^2\beta_t^2dM_t^2,
    \end{aligned}
\end{align}

where $\phi_t'\sigma_t = \sum_{i=0}^2\phi_t^i\sigma_t^i$, and
\begin{equation*}
    \Theta_t^0 = \frac{\mu_t^0-r_t}{\sigma_t^0},\quad\Theta_t^1=\frac{\mu_t^1-r_t-\sigma_t^1\Theta_t^0}{\beta_t^1\lambda_t^1}\mathbbm{1}_{\{\beta_t^1\lambda_t^1\ne0\}},\quad\Theta_t^2=\frac{\mu_t^2-r_t-\sigma_t^2\Theta_t^0}{\beta_t^2\lambda_t^2}\mathbbm{1}_{\{\beta_t^2\lambda_t^2\ne0\}}.
\end{equation*}

\begin{assumption}\label{assumption:LinearCompleteMarket}
    We assume that the processes $\Theta^0$, $\Theta^1\sqrt{\lambda^1}$ and $\Theta^2\sqrt{\lambda^2}$ are bounded.
\end{assumption}

Let $T>0$. Let $\eta\in L^2(\mathcal{G}_T)$ and let $D$ be a finite variational optional process in $\mathcal{A}_T^2$. We consider a European option with terminal time $T$ which generates a terminal payoff $\eta$ and intermediate cashflows, commonly referred to as `dividends' (which need not be strictly positive, c.f., e.g., \cite{crepey2015bilateral}). For each $t\in[0,T]$, $D_t$ represents the cumulative intermediate cashflows generated by the option between $[0,t]$. As the  `dividends' are not necessarily positive, the process $D$ is not necessarily non-decreasing.

We place ourselves from the point of view of an agent who wants to sell this option at time $t=0$. With the proceeds from the sale, they wish to construct a (self-financing)  portfolio which allows them to pay the buyer of the contract the amount $\eta$ at time $T$ as well as the intermediate `dividends' $D$.

Setting $Z_t\coloneqq\phi_t'\sigma_t$ and $K_t^i=\phi_t^i\beta_t^i$, for $i\in\{1,2\}$, and using \eqref{eq:3_SFC} (with $C=D$), we get that the process $(V, Z, K^1, K^2)$ satisfies the following dynamics,
\begin{equation}\label{eq:3_BSDE}
    -dV_t = -(r_tV_t+\Theta_t^0Z_t + \Theta_t^1K_t^1\lambda_t^1+\Theta_t^2K_t^2\lambda_t^2)dt +dD_t - Z_tdW_t - K_t^1dM_t^1 - K_t^2dM_t^2.
\end{equation}
For each $(\omega, t, y, z, k^1, k^2)$ we define,
\begin{equation}\label{eq:3_LinearDriver}
    g(\omega,t,y,z,k^1,k^2) \coloneqq -r_t(\omega)y - \Theta_t^0(\omega)z-\Theta_t^1(\omega)k^1\lambda_t^1(\omega)-\Theta_t^2(\omega)k^2\lambda_t^2(\omega).
\end{equation}
By our previous assumptions  and Assumption \ref{assumption:LinearCompleteMarket}, we have that $r,\Theta^0$ and $\Theta^i\sqrt{\lambda^i}$ (for $i\in\{1,2\}$) are predictable and bounded. It  follows that the driver $g$ is a $\lambda^{(p)}$-linear driver (cf. Eq. \eqref{eq:defLambdaLinear}). By Proposition \ref{prop:EU},  there exists a unique solution $(X,Z,K^1,K^2)\in\mathcal{S}^2\times\mathcal{H}^2\times\mathcal{H}_{\lambda^1}^2\times\mathcal{H}_{\lambda^2}^2$ of the BSDE associated with terminal time $T$, generalized $\lambda^{(p)}$-linear driver $g(t,y,z,k^1,k^2)dt+dD_t$ and terminal condition $\eta\in L^2(\mathcal{G}_T)$.

The solution of the BSDE $(X,Z,K^1,K^2)$ provides a replicating portfolio, where the seller chooses a risky-asset strategy $\phi$ according to the following change of variables: 
\begin{equation}\label{eq:3_LinearStrategy}
    \mathbf{\Phi}:\mathcal{H}_T^2\times\mathcal{H}_{\lambda^1,T}^2\times\mathcal{H}_{\lambda^2,T}^2\rightarrow\mathcal{H}_T^2\times\mathcal{H}_{\lambda^1,T}^2\times\mathcal{H}_{\lambda^2,T}^2; (Z,K^1, K^2)\mapsto\mathbf{\Phi}(Z,K^1,K^2)\coloneqq\phi,
\end{equation}
where $\phi\coloneqq(\phi^0,\phi^1,\phi^2)$, and the amount $\phi^i$ invested in the $i$-th asset (where $i\in\{0,1,2\}$) is,
\begin{equation}\label{eq:3_VarChange}
    \phi_t^0 = \frac{Z_t-\frac{K_t^1\sigma_t^1}{\beta_t^1}-\frac{K_t^2\sigma_t^2}{\beta_t^2}}{\sigma_t^0},\quad\phi_t^1=\frac{K_t^1}{\beta_t^1},\quad\phi_t^2=\frac{K_t^2}{\beta_t^2}.
\end{equation}

Here, the process $D$ corresponds to the cumulative cash `withdrawn' by the seller from their hedging (replicating) portfolio. The above portfolio is a replicating portfolio for the seller of the European contingent claim, since the seller is able to reinvest all proceeds from the sale into the market and pay $\eta$ at the option expiration date of $T$, as well as the intermediate `dividends' of the option.

The amount $X_0$ (the first component of the BSDE  at time zero) is the hedging price (or price by replication) of the option at time $t=0$ and we denote it with $X_{0,T}(\eta, D)$. For $t\in[0,T]$, the hedging price (or price by replication) $X_t$ is denoted by $X_{t,T}(\eta, D)$.

\subsubsection{The Case Where $D$ is Predictable}

Let the cumulative `dividend' process $D$ be a (\textbf{predictable}) process  in  $\mathcal{A}_{p}^2$. Since the driver from \eqref{eq:3_LinearDriver} is $\lambda^{(p)}$-\textbf{linear}, we have, by Theorem \ref{theorem:RRGLB}, an explicit formula for $X_{t,T}(\eta,D)$. For each $t\in[0,T]$ the adjoint process $(\Gamma_{t,s})_{s\in[t,T]}$ is the unique solution of the following SDE,
\begin{equation*}\label{eq:3_AdjointProcess}
    d\Gamma_{t,s} = \Gamma_{t,s-}(-r_sds - \Theta_s^0dW_s - \Theta_s^1dM_s^1 - \Theta_s^2dM_s^2),\quad\Gamma_{t,t} = 1.
\end{equation*}
By Remark \ref{theorem:RRGLB}, $(\Gamma_{t,s})_{s\in[t,T]}$ is written,
\begin{align*}\label{eq:3_AdjointProcess2}
    \begin{aligned}
        \Gamma_{t,s} &= \exp\left(-\int_t^sr_udu\right)\exp\left(-\int_t^s\Theta_u^0dW_s -\frac{1}{2}\int_t^s(\Theta_u^0)^2-\Theta_u^1\lambda_u^1-\Theta_u^2\lambda_u^2 du\right)\\&\quad\times(1-\Theta_{\tau_1}^1\mathbbm{1}_{\{t<\tau_1\le s\}})(1-\Theta_{\tau_2}^2\mathbbm{1}_{\{t<\tau_2\le s\}})\\
        &=e^{-\int_t^sr_udu}\zeta_{t,s},
    \end{aligned}
\end{align*}
where the process $(\zeta_{t,s})_{s\in[t,T]}$ satisfies the dynamics,
\begin{equation}\label{eq:3_Zeta1}
    d\zeta_{t,s} = \zeta_{t,s-}(-\Theta_s^0dW_s - \Theta_s^1dM_s^1-\Theta_s^2dM_s^2),\quad\zeta_{t,t}=1.
\end{equation}

Hence, by the explicit formula (cf. Theorem \ref{theorem:RRGLB} where $D$ is predictable),
\begin{equation}\label{eq:3_RepresentationP1}
    X_{t,T}(\eta,D) = X_t = \mathbb{E}\left[e^{-\int_t^Tr_sds}\zeta_{t,T}\eta + \int_t^Te^{-\int_t^sr_udu}\zeta_{t,s-}dD_s\middle|  \mathcal{G}_t\right].
\end{equation}

\subsubsection{The Case Where $D$ is Optional}

Let us now consider the case where $D$ is not necessarily predictable, but only optional; more precisely,  $D\in\mathcal{A}^2$. By Proposition \ref{prop:OptionalpJumps}, there exist a unique process $D'\in\mathcal{A}_p^2$ and unique processes $\psi^1\in\mathcal{H}_{\lambda^1}^2,\psi^2\in\mathcal{H}_{\lambda^2}^2$, such that for all $t\in[0,T]$,
\begin{equation}\label{eq:3_DividendProcess}
    D_t = D_t' + \int_0^t\sum_{i=1}^2\psi_s^idN_s^i,\quad\text{a.s.}
\end{equation}
From a financial point of view, the random variable $\psi_{\tau_i}^i$ represents the cash flow generated by the contingent claim (the option) at the $i$-th default time $\tau_i$ (see also \cite{bielecki2004hedging} Part I for contingent claims where the cash flow depends on the default times). By Theorem \ref{theorem:RRGLB2}, the hedging price at time $t$, $X_{t,T}(\eta,D)$, is equal to:
\begin{multline}\label{eq:3_RepresentationP2}
    X_{t,T}(\eta,D) = \mathbb{E}\left[e^{-\int_t^Tr_sds}\zeta_{t,T}\eta \right.+ \int_t^Te^{-\int_t^sr_udu}\zeta_{t,s-}dD_s'\\\left.+e^{-\int_t^{\tau_1}r_sds}\zeta_{t,\tau_1}\psi_{\tau_1}^1\mathbbm{1}_{\{t<\tau_1\le T\}}+e^{-\int_t^{\tau_2}r_sds}\zeta_{t,\tau_2}\psi_{\tau_2}^2\mathbbm{1}_{\{t<\tau_2\le T\}}\middle|\mathcal{G}_t\right].
\end{multline}

\subsubsection{Change of measure}
The change of measure technique is often used in linear market models in financial mathematics. In this sub-subsection, we will make the following assumption on the `Sharpe ratios' $\Theta^0$, $\Theta^1$ and $\Theta^2$.

\begin{assumption}\label{assumption:PriceOfRiskProduct}
	We assume that $\prod_{i=1}^2(1 - \Theta_{\tau_i}^i\mathbbm{1}_{\{t<\tau_i\le s\}})>0$ for all $0\le t\le s\le T$ a.s.
\end{assumption}
By Remark \ref{remark:2.12} and Assumption \ref{assumption:PriceOfRiskProduct} we have that $\zeta_{t,s}>0$ for all $s\in[t,T]$. By Assumption \ref{assumption:LinearCompleteMarket} we have that $\int_0^T((\Theta_s^0)^2 + (\Theta_s^1)^2\lambda_s^1 + (\Theta_s^2)^2\lambda_s^2)ds$ is bounded; hence, by Proposition \ref{prop:ESM-M} $(\zeta_{t,s})_{s\in[t,T]}$ is a square-integrable martingale. Let $Q$ be a new probability measure, defined by the Radon-Nikodym derivative with respect to $P$ on $\mathcal{G}_T$:
\begin{equation}\label{eq:3_RNd}
    \left.\frac{dQ}{dP}\right|_{\mathcal{G}_T} = \zeta_{0,T} = \mathcal{E}\left(-\int_0^\cdot\Theta_s^0dW_s - \int_0^\cdot\Theta_s^1dM_s^1 - \int_0^\cdot\Theta_s^2dM_s^2\right)_T.
\end{equation}

\paragraph{The Case Where $D$ is Predictable}
In the case where the `dividend' process $D$ is \textbf{predictable}, we  use Bayes formula to perform a change of measure in the conditional expectation of \eqref{eq:3_RepresentationP1} (cf., e.g.,  Proposition 1.7.1.5 from \cite{jeanblanc2009mathematical}) to get:
\begin{equation}\label{eq:3_RPQ}
    X_{t,T}(\eta,D) = X_t = \mathbb{E}^Q\left[e^{-\int_t^Tr_sds}\eta + \int_t^Te^{-\int_t^sr_udu}dD_s\middle|\mathcal{G}_t\right]\quad\text{a.s.}
\end{equation}
\paragraph{The Case Where $D$ is Optional}
If the `dividend' process $D$ is \textbf{optional}, the price of the European option at time $t$ under the probability measure $Q$ can be written:
\begin{multline}\label{eq:3_RepresentationP3}
    X_{t,T}(\eta,D) = \mathbb{E}^Q\left[e^{-\int_t^Tr_sds}\eta \right.+ \int_t^Te^{-\int_t^sr_udu}dD_s'\\\left.+e^{-\int_t^{\tau_1}r_sds}\psi_{\tau_1}^1\mathbbm{1}_{\{t<\tau_1\le T\}}+e^{-\int_t^{\tau_2}r_sds}\psi_{\tau_2}^2\mathbbm{1}_{\{t<\tau_2\le T\}}\middle|\mathcal{G}_t\right]
\end{multline}

We note that the pricing system for this market is linear.

\subsection{\texorpdfstring{Pricing in a Non-linear Complete Market with $p$ Defaultable Assets}{Pricing in a Non-linear Complete Market with p Defaultable Assets}}\label{subsec:non-linearPricingCompleteMarket}

We now assume that there are imperfections in the market, which are incorporated via the non-linearity of the driver in the dynamics of the wealth process. We consider the case where there are $p$ defaultable assets.

We introduce the following notation for the price processes of the primary assets: $B, S^0, S^1, \ldots, S^p$, where $B$ and $S^0$ represent the price process of a risk-free savings account and a default-free risky asset, respectively, while for each $i \in \{1, \ldots, p\}$, $S^i$ is the price process of the $i$-th defaultable asset (or $i$-th credit risk bearing asset).  The underlying probabilistic framework is the same as that introduced at the beginning of Section \ref{sec:PPS}, and we continue to work under the same assumptions. The price processes of $B$ and $S^0$ remain unchanged from \eqref{eq:3.1}. For each $i \in \{1, \ldots, p\}$, the price process of the $i$-th defaultable asset is given by
\begin{equation}\label{eq:3_SPriceProcess1}
    dS_t^i = S_{t-}^i\left[\mu_t^i \, dt + \sigma_t^i \, dW_t + \beta_t^i \, dM_t^i\right], \quad S_0^i > 0.
\end{equation}
For each $i \in \{1, \ldots, p\}$, the processes $\mu^i$ and $\sigma^i$ are assumed to be predictable, with $\sigma^i > 0$, and such that $\mu^i$, $\sigma^i$, and $(\sigma^i)^{-1}$ are bounded. The interest rate process $r$ is assumed to be predictable and bounded. We assume that, for each $i \in \{1, \ldots, p\}$, $\beta_t^i\neq 0$.  We recall that, for each $i \in \{1, \ldots, p\}$, the process $M^i = N^i - \int_0^\cdot \lambda_s^i \, ds$ is the $\mathbb{G}$-compensated default martingale.

We again consider an investor who, at time $0$, invests an initial amount $x \in \mathbb{R}$ in this market. For $i \in \{1, \ldots, p\}$, we let $\phi_t^i$ denote the amount of money invested in $S_t^i$ at time $t \in [0, T]$. If, for a given $i \in \{1, \ldots, p\}$, we have $\beta_{\tau_i}^i = -1$ a.s., then on the set $\{t \ge \tau_i\}$, the price of the $i$-th asset becomes $0$, and hence the investor will no longer invest in this asset. If $\beta_{\tau_i}=-1$, we set $\phi_t^i = 0$ on $\{t \ge \tau_i\}$.

Similarly to the linear framework, for a given risky-asset strategy denoted $\phi = (\phi^0,\phi^1,\ldots,\phi^p) \in \mathcal{H}^2 \times \mathcal{H}_{\lambda^1}^2 \times \cdots \times \mathcal{H}_{\lambda^p}^2$, a given cash withdrawal finite variation optional process $C \in \mathcal{A}_T^2$, and a given initial wealth (capital) $x \in \mathbb{R}$, the wealth process at time $t \in [0,T]$, denoted by $V_t^{x,\phi,C}$ (or simply $V_t$ if there is no ambiguity), satisfies the self-financing condition:
\begin{equation}\label{eq:3_SelfFinancingnon-linear1}
   -dV_t = g(t,V_t, \phi_t'\sigma_t, \phi_t^1\beta_t^1,\ldots,\phi_t^p\beta_t^p)dt - \phi_t'\sigma_t\,dW_t - \sum_{i=1}^p \phi_t^i \beta_t^i\,dM_t^i + dC_t; \quad V_0 = x,
\end{equation}
where $g$ is a possibly non-linear $\lambda^{(p)}$-admissible driver. Equivalently, setting $Z_t = \phi_t'\sigma_t$ and, for each $i \in \{1, \ldots, p\}$, $K_t^i = \phi_t^i\beta_t^i$, we have
\begin{equation}\label{eq:3_SelfFinancingnon-linear2}
   -dV_t = g(t,V_t, Z_t, K_t^1,\ldots,K_t^p)dt - Z_t\,dW_t - \sum_{i=1}^p K_t^i\,dM_t^i + dC_t; \quad V_0 = x.
\end{equation}

We consider a European contingent claim with maturity $T$, terminal payoff  $\eta \in L^2(\mathcal{G}_T)$, and optional 'dividend' process $D \in \mathcal{A}_T^2$. \\Let $(X_{\cdot,T}^g(\eta,D), Z_{\cdot,T}^g(\eta,D), K_{\cdot,T}^{g,1}(\eta,D), \ldots, K_{\cdot,T}^{g,p}(\eta,D))$, or simply $(X, Z, K^1, \ldots, K^p)$, denote the solution of the BSDE with terminal time $T$, terminal condition $\eta$, and generalized driver $g(t,y,z,k^1,\ldots,k^p)\,dt + dD_t$, that is, the BSDE satisfying the following dynamics:
\begin{equation}\label{eq:3_BSDEnon-linear1}
    -dX_t = g(t,X_t,Z_t,K_t^1,\ldots,K_t^p)dt + dD_t - Z_t\,dW_t - \sum_{i=1}^p K_t^i\,dM_t^i, \quad X_T = \eta.
\end{equation}

Hence, the process $X = X_{\cdot,T}^g(\eta,D)$ coincides with the wealth process corresponding to initial wealth $x = X_0$, cumulative cash withdrawal process $D$, and risky-asset strategy $\phi = \mathbf{\Phi}(Z,K^1,\ldots,K^p)$, where $\mathbf{\Phi}$ is the following generalization of \eqref{eq:3_LinearStrategy}:
\begin{align*}
    \mathbf{\Phi} : \mathcal{H}^2 \times \mathcal{H}_{\lambda^1}^2 \times \cdots \times \mathcal{H}_{\lambda^p}^2 &\rightarrow \mathcal{H}^2 \times \mathcal{H}_{\lambda^1}^2 \times \cdots \times \mathcal{H}_{\lambda^p}^2, \\
    (Z, K^1, \ldots, K^p) &\mapsto \mathbf{\Phi}(Z, K^1, \ldots, K^p) \coloneqq \phi,
\end{align*}
where $\phi\coloneqq(\phi^0,\phi^1,\ldots,\phi^p)$ and the amount $\phi^i$ invested in the $i$-th asset (where $i\in\{0,1,\ldots,p\}$) is:
\begin{equation}
    \phi_t^0 = \frac{Z_t-\frac{K_t^1\sigma_t^1}{\beta_t^1}-\frac{K_t^2\sigma_t^2}{\beta_t^2}}{\sigma_t^0},\quad\phi_t^i=\frac{K_t^i}{\beta_t^i},\quad\text{for }i\in\{1,\ldots,p\}.
\end{equation}
Thus,  $X = V^{X_0,\phi,D}$.

Starting from the initial wealth $X_0 = X_{0,T}^g(\eta,D)$, the seller can construct a risky-asset strategy $\phi$ that allows them to pay the intermediate 'dividends' $D$ and the final payoff $\eta$. We therefore call the initial wealth $X_0$ the hedging price (or replicating price) at time $t = 0$ of the option, and the process $\phi$ the hedging strategy (or replicating strategy).

More generally, let us  consider a maturity time $S \in [0,T]$. For each $S \in [0,T]$ and for each payoff-dividend pair $(\eta,D) \in L^2(\mathcal{G}_S) \times \mathcal{A}_S^2$, the process $X_{\cdot,S}^g(\eta,D)$ is called the hedging price of the option with maturity $S$ and payoff-dividend pair $(\eta,D)$. This yields the following pricing system for the $p$-defaultable non-linear market model:
\begin{equation}\label{eq:3_non-linearPricingSystem1}
    \mathbf{X}^g : (S, \eta, D) \mapsto X_{\cdot,S}^g(\eta,D),
\end{equation}
which is generally non-linear with respect to the pair $(\eta,D)$ (as the driver $g$ is in general non-linear).

We now state some properties of this  pricing system (cf.  \cite{dumitrescu2018bsdes} for the case of single default, and  \cite{el1997non} for the case without jumps).

 \subsubsection{\texorpdfstring{Properties of the Non-linear Pricing System  $\mathbf{X}^g$  in the Case of $p$ Defaultable Assets}{Properties of the Non-linear Pricing System Xg in the Case of p Defaultable Assets}}\label{subsec:PropertiesofPricingSystem}

 \begin{itemize}
     \item \textbf{Consistency:} By the flow property of the BSDEs with default jumps, the pricing system $\mathbf{X}^g$ is consistent. More precisely, let  $S'\in[0,T]$, $\eta\in L^2(\mathcal{G}_{S'})$, $D\in\mathcal{A}_{S'}^2$ and $S\in[0,S']$. Then the hedging price of the European contingent claim associated with terminal payoff $\eta$, cumulative dividend process $D$ and maturity $S'$ coincides with the hedging price of the European option with terminal time $S$, payoff $X_{S,S'}^g(\eta,D)$ and dividend process $(D_t)_{t\le S}$, that is,
     \begin{equation*}
         X_{\cdot, S'}^g(\eta,D) = X_{\cdot,S}^g\left(X_{S,S'}^g(\eta,D),D\right).
     \end{equation*}
     \begin{remark}\label{remark:3_g00}
         Note that when $g(t,0,0,\ldots,0)=0$, we get that the price of a European option with null payoff and no dividends is equal to $0$ for all maturity times $S\in[0,T]$, hence $X_{\cdot,S}^g(0,0)=0$.
     \end{remark}
 \end{itemize}

Due to the (possible) presence of default jumps the non-linear pricing system is not necessarily monotone with respect to the payoff and dividend process. We introduce the following assumption (cf.  the comparison Theorem \ref{thm:CompThm}).

 \begin{assumption}\label{assumption:JumpPricingSystemMap}
     We assume that for each $i\in\{1,\ldots,p\}$ there exists a map,
     \begin{equation*}
         \gamma^i:\Omega\times[0,T]\times\mathbb{R}^4\rightarrow\mathbb{R};\text{  }(\omega,t,y,z,k^i,\hat{k}^i)\mapsto\gamma_t^{y,z,k^i,\hat{k}^i}(\omega)
     \end{equation*}
     which is $\mathcal{P}\otimes\mathcal{B}(\mathbb{R}^4)$-measurable, satisfying $dP\otimes dt$-a.e. for each $(y,z,k^i,\hat{k}^i)\in\mathbb{R}^4$,
     \begin{equation}\label{eq:3_MapCondition1}
         \left|\gamma_t^{y,z,k^i,\hat{k}^i}\sqrt{\lambda_t^i}\right|\le C\quad\text{and}\quad\gamma_t^{y,z,k^i,\hat{k}^i}\ge-1,
     \end{equation}
     and such that
     \begin{equation}\label{eq:3_MapCondition2}
         g(t,y,z,k^1,\ldots,k^p) - g(t,y,z,\hat{k}^1,\ldots,\hat{k}^p) \ge \sum_{i=1}^p\gamma_t^{y,z,k^i,\hat{k}^i}(k^i-\hat{k}^i)\lambda_t^i.
     \end{equation}
 \end{assumption}


 We now introduce the following partial order relation (cf. also \cite{dumitrescu2018bsdes}). Let $S\in[0,T]$ be given. For $(\eta,D),(\hat{\eta},\hat{D})\in L^2(\mathcal{G}_S)\times\mathcal{A}_S^2$, we say that $(\eta,D)$ dominates $(\hat{\eta},\hat{D})$ and we write the following relation,
 \begin{equation*}
    (\eta,D)\succ(\hat{\eta},\hat{D})\quad \text{if}\quad \eta\ge\hat{\eta}\text{ a.s.  and }D - \hat{D}\text{ is non-decreasing}.
 \end{equation*}

\begin{proposition}\label{proposition:Non-linearPricingSystemProperties}
	Under Assumption \ref{assumption:JumpPricingSystemMap}, the non-linear pricing system $\mathbf{X}^g$ has the following properties:
	\begin{enumerate}[label=(\alph*)]
		\item \textbf{Monotonicity:} The non-linear pricing system $\mathbf{X}^g$ is non-decreasing with respect to the payoff-dividend pair. More precisely, for all maturity times $S\in[0,T]$, for all payoffs $\eta,\hat{\eta}\in L^2(\mathcal{G}_S)$ and all cumulative dividend processes $D,\hat{D}\in\mathcal{A}_S^2$, the following implication  holds:
		    \begin{equation*}
		        \text{If }(\eta,D)\succ(\hat{\eta},\hat{D}),\text{ then we have }X_{t,S}^g(\eta,D)\ge X_{t,S}^g(\hat{\eta},\hat{D}),\quad t\in[0,S]\text{ a.s.}
		    \end{equation*}
		\item \textbf{Convexity:} If $g$ is convex with respect to the vector $(y,z,k^1,\ldots,k^p)$, then the non-linear pricing system $\mathbf{X}^g$ is convex with respect to the payoff-dividend pair $(\eta,D)$, that is, for any $\alpha\in[0,1]$, $S\in[0,T]$, $\eta,\hat{\eta}\in L^2(\mathcal{G}_S)$ and $D,\hat{D}\in\mathcal{A}_S^2$, we have: for all $t\in[0,S]$,
		    \begin{equation*}
		        X_{t,S}^g(\alpha\eta + (1-\alpha)\hat{\eta}, \alpha D + (1-\alpha)\hat{D})\le \alpha X_{t,S}^g(\eta,D) + (1-\alpha)X_{t,S}^g(\hat{\eta},\hat{D})\quad\text{a.s.}
		    \end{equation*}
		\item \textbf{Non-Negativity:} When $g(t,0,0,0,\ldots,0)=0$, the non-linear pricing system $\mathbf{X}^g$ is non-negative, that is, for all $S\in[0,T]$, for all non-negative terminal payoffs $\eta\in L^2(\mathcal{G}_S)$ and for all non-decreasing optional dividend processes $D\in\mathcal{A}_S^2$, we have that $X_{t,S}^g(\eta,D)\ge0$ for all $t\in[0,S]$ a.s.
		\item \textbf{No Arbitrage:} We assume  the additional condition that,  for each $i\in\{1,\ldots,p\}$, $\gamma_{t}^{y,z,k^i,\hat{k}^i}>-1,$ $dP\otimes dt$-a.e.  Then,  the non-linear pricing system $\mathbf{X}^g$ satisfies the no arbitrage property. That is, for all maturities $S\in[0,T]$, for all terminal payoffs $\eta,\hat{\eta}\in L^2(\mathcal{G}_S)$ and for all optional cumulative dividend processes $D,\hat{D}\in\mathcal{A}_S^2$, the following holds:
    If $(\eta,D)\succ(\hat{\eta},\hat{D})$ and if at time $t_0\in[0,S]$ we have $X_{t_0,S}^g(\eta,D)=X_{t_0,S}^g(\hat{\eta},\hat{D})$ a.s., then $\eta=\hat{\eta}$ a.s. and $(D_t-\hat{D}_t)_{t_0\le t\le S}$ is constant.
	\end{enumerate}
\end{proposition}

\begin{proof}For  the \textbf{monotonicity} of the non-linear pricing system $\mathbf{X}^g$ we use the comparison result from Theorem \ref{thm:CompThm} (which is applicable under  Assumption \ref{assumption:JumpPricingSystemMap}) and set $g=\hat{g}$ in this theorem.\\ 
		 For the \textbf{ convexity} of the non-linear pricing system $\mathbf{X}^g$,  we use again the comparison result from Theorem  \ref{thm:CompThm}. The proof follows standard arguments.\\
		The \textbf{non-negativity} is a direct  consequence of the monotonicity property. If $\hat{\eta} = 0$, $\hat{D}=0$, $\eta\ge0$ a.s., and $D$ non-decreasing, then by the definition of the partial order relation $\succ$ we have, $(\eta,D)\succ(\hat{\eta},\hat{D})$. By the comparison result from Theorem \ref{thm:CompThm} and Remark \ref{remark:3_g00} we have that $X_{t,S}^g(\eta,D)\ge X_{t,S}^g(\hat{\eta},\hat{D})=X_{t,S}^g(0,0)=0$.\\
 The proof of the \textbf{no arbitrage property} is a direct consequence of  the strict comparison result from Theorem \ref{thm:CompThm}, where we set $\hat{g}=g$.
\end{proof}

\subsubsection{\texorpdfstring{The $(g,D)$-Conditional Evaluation $\mathscr{E}^{g,D}$  for a $\lambda^{(p)}$-Admissible Driver}{The (g,D)-Conditional Evaluation EgD and EgD-Martingale for a lambda-p-Admissible Driver g}}\label{subsec:ThegdConditionalExpectation}

Let $g$ be a $\lambda^{(p)}$-admissible driver and let $D$ be a given  optional dividend process belonging to $\mathcal{A}_T^2$. For each $S\in[0,T]$ and each $\eta\in L^2(\mathcal{G}_S)$, we define the $(g,D)$-conditional evaluation of $\eta$ by,
\begin{equation*}
    \mathscr{E}_{t,S}^{g,D}(\eta)\coloneqq X_{t,S}^g(\eta,D),\quad0\le t\le S.
\end{equation*}
The $(g,D)$-conditional evaluation $\mathscr{E}_{\cdot,S}^{g,D}(\eta)$ is the first component of the solution of the BSDE associated with terminal time $S$, generalized driver $g(t,y,z,k^1,\ldots,k^p)dt + dD_t$ and terminal condition $\eta$, where we have fixed $D$ in the space $\mathcal{A}_T^2$.

\begin{remark}
    In the case where $D=0$, our $(g,D)$-conditional evaluation reduces to the $g$-conditional evaluation for the case of $p$ default times (which we denote by  $\mathscr{E}^g$).
    If, furthermore,  $g=0$, then the $(g,D)$-conditional evaluation reduces to the standard conditional expectation under $P$, that is $\mathscr{E}_{t,S}^{0,0}=\mathbb{E}^P[\eta|\mathcal{G}_t]$, for $t\in[0,S]$.
\end{remark}

\begin{remark}
    Note that we can in fact define the $(g,D)$-conditional evaluation $\mathscr{E}_{\cdot,S}^{g,D}(\eta)$ on the entire interval $[0,T]$ by setting,
    \begin{equation*}
        \mathscr{E}_{t,S}^{g,D}(\eta)\coloneqq\mathscr{E}_{t,T}^{g^S,D^S}(\eta)\text{ for }t\ge S,
    \end{equation*}
where we have set $g^S(t,\cdot)\coloneqq g(t,\cdot)\mathbbm{1}_{t\le S}$ and $D_t^S\coloneqq D_{t\wedge S}$.
\end{remark}

Let now $\mathcal{T}_0$ be the set of all stopping times. We extend the definition of the $(g,D)$-conditional evaluation for each terminal stopping time $\tau\in\mathcal{T}_0$ and each $\eta\in L^2(\mathcal{G}_\tau)$ as the first component of the solution of the BSDE associated with  terminal time $T$, $\lambda^{(p)}$-admissible driver $g^\tau(t,\cdot)\coloneqq g(t,\cdot)\mathbbm{1}_{t\le \tau}$ and optional process $D_t^\tau\coloneqq D_{t\wedge\tau}$.

Some properties of the non-linear $(g,D)$-evaluations are as follows (cf.  \cite{dumitrescu2018bsdes} for the single default jump case):   

\begin{itemize}
    \item \textbf{Consistency:} Let $\tau,\tau'\in\mathcal{T}_0$ be such that $\tau\le\tau'$ a.s. and let $\eta\in L^2(\mathcal{G}_{\tau'})$. Then, $\mathscr{E}_{t,\tau'}^{g,D}(\eta) = \mathscr{E}_{t,\tau}^{g,D}(\mathscr{E}_{\tau,\tau'}^{g,D}(\eta))$ a.s.
    \item \textbf{Generalized Zero-One Law:} Let $\tau\in\mathcal{T}_0$, let  $\eta\in L^2(\mathcal{G}_\tau)$. For $t\in[0,T]$ and for  $A\in\mathcal{F}_t,$ we have,
    \begin{equation*}
        \mathscr{E}_{t,\tau}^{g^A,D^A}(\mathbbm{1}_{A}\eta) = \mathbbm{1}_{A}\mathscr{E}_{t,\tau}^{g,D}(\eta)\text{ a.s.,}
    \end{equation*}
    where $g^A(s,\cdot)\coloneqq g(s,\cdot)\mathbbm{1}_A\mathbbm{1}_{(t,T]}(s)$ and $D_s^A\coloneqq(D_s-D_t)\mathbbm{1}_A\mathbbm{1}_{s\ge t}.$   In the case where $D=0$, this property has been established  in \cite{grigorova2017optimal} (in the case of a Brownian-Poisson filtration). 
    \item \textbf{Monotonicity:} Using the comparison theorem (Theorem \ref{thm:CompThm}), under Assumption \ref{assumption:JumpPricingSystemMap}, the $(g,D)$-conditional evaluation $\mathscr{E}^{g,D}(\cdot)$ is monotone with respect to the terminal payoff. 
    \item \textbf{Convexity:} Under Assumption \ref{assumption:JumpPricingSystemMap}, if we further assume that $g$ is  convex  with respect to the vector $(y,z,k^1,\ldots,k^p)$, then $\mathscr{E}^{g,D}(\cdot)$ is convex with respect to the terminal payoff. 
    \item \textbf{No Arbitrage Property:} Under Assumption \ref{assumption:JumpPricingSystemMap}, if we  further assume that for each $i\in\{1,\ldots,p\}$, $\gamma_{t}^{y,z,k^i,\hat{k}^i}>-1,$ $dP\otimes dt$-a.e.,  then by the strict comparison Theorem (Theorem \ref{thm:CompThm} (ii)), $\mathscr{E}^{g,D}$ has the no arbitrage property.
\end{itemize}

We now present two examples.  

\subsection{\texorpdfstring{Example: Large Seller who Affects the $i$-th Default Probability}{Example: Large Seller who Affects the ith Default Probability}}\label{subsection_large1}

We place ourselves in the same probabilistic framework as in Subsection \ref{subsec:non-linearPricingCompleteMarket}. We consider a European option with maturity $T>0$, terminal payoff $\eta\in L^2(\mathcal{G}_T)$ and an optional dividend process $D\in\mathcal{A}_T^2$. We consider the situation where  the seller of this European option is a large trader.  The hedging strategy of the trader (and its associated wealth process) may affect  the  default probabilities of the assets. For this example we assume that the large seller only affects the $i$-th default probability (where $i$ is a fixed index). We also assume that the $i$-th default intensity is bounded. The large seller takes this feedback effect into consideration for their market model.

 Let $i$ be a \textbf{fixed} index in $\{1,\ldots,p\}$ (in the whole sub-section). We are given a family of probability measures parametrized by  the risky-asset strategy $\phi$ and the (self-financing)  wealth process $V$. More precisely, let  $\phi\in\mathcal{H}^2\times\mathcal{H}_{\lambda^1}^2\times\cdots\times\mathcal{H}_{\lambda^p}^2$ and let $V\in\mathcal{S}^2$.  Let $Q^{V,\phi,i}$ be the probability measure  defined by the Radon-Nikodym density process  (with respect to $P$):
\begin{equation}\label{eq:3_RadonNikodymExample1}
    \left.\frac{dQ^{V,\phi,i}}{dP}\right|_{\mathcal{G}_t} = L_t^{V,\phi,i},
\end{equation}
where $L^{V,\phi,i}$ is the solution of the SDE,
\begin{equation}\label{eq:3_DensitySDE1}
    dL_t^{V,\phi,i} = L_{t-}^{V,\phi,i}\gamma^i(t,V_{t-},\phi_t)dM_t^i;\quad L_0^{V,\phi,i}=1.
\end{equation}

We introduce  the following assumption on  the function $\gamma^i$. 

\begin{assumption}\label{assumption:Gammai}
  The function  $\gamma^i:(\omega, t, y, \phi^0,\phi^1,\ldots,\phi^p)\mapsto\gamma^i(\omega,t,y,\phi^0,\phi^1,\ldots,\phi^p)$ is a $\mathcal{P}\otimes\mathcal{B}(\mathbb{R}^{p+2})$-measurable function defined on $\Omega\times[0,T]\times\mathbb{R}^{p+2}$, bounded, and such that the map $y\mapsto\gamma^i(\omega,t,y,\phi^0,\phi^1,\ldots,\phi^p)/\phi^i$ is uniformly Lipschitz. In addition we  assume that $\gamma^i(t,\cdot)>-1$, $dt\otimes dP$-a.e. 
\end{assumption}

In the financial context, we use  the function $\gamma^i$ to represent the influence of the seller's strategy on the default intensity of the $i$-th asset, where $\phi$ is the seller's risky-asset strategy and $V$ is the value of their portfolio.

By Assumption \ref{assumption:Gammai}, Remark \ref{remark:2.12} and Proposition \ref{prop:ESM-M}, the process $L^{V,\phi,i}$ is positive and belongs to $\mathcal{S}^2$.

Using  Girsanov's theorem (and our assumptions on $\lambda_i$ and $\gamma^i$)  the process   $(W_t^{V,\phi,i})$ hereafter is a $Q^{V,\phi,i}$-Brownian motion and the process $(M_t^{V,\phi,i})$ is a $Q^{V,\phi,i}$- martingale, where 
\begin{align}
    &\begin{aligned}
        W_t^{V,\phi,i} &\coloneqq W_t - \int_0^t\frac{d\langle W,L^{V,\phi,i}\rangle_s}{L_{s-}^{V,\phi,i}}= W_t;
    \end{aligned}\label{eq:3_WQMartingale1}\\
    &\begin{aligned}
        M_t^{V,\phi,i} &\coloneqq M_t^i - \int_0^t\frac{d\langle M^i,L^{V,\phi,i}\rangle_s}{L_{s-}^{V,\phi,i}}
        = M_t^i - \int_0^t\gamma^i(s,V_{s-},\phi_s)\lambda_s^ids
    \end{aligned}\label{eq:3_MQLocalMartingale1}
\end{align}

 Hence under $Q^{V,\phi,i}$, the $i$-th $\mathbb{G}$-default intensity process is equal to $\lambda_t^i(1+\gamma^i(t,V_{t-},\phi_t))$ since we can rewrite \eqref{eq:3_MQLocalMartingale1} as,
\begin{equation*}
    M_t^{V,\phi,i}\coloneqq M_t^i - \int_0^t\gamma^i(s,V_{-},\phi_s)\lambda_s^ids=N_t^i-\int_0^t\lambda_s^i(1+\gamma^i(s,V_{s-},\phi_s))ds.
\end{equation*}

\begin{remark}\label{remark:3_Girsanov}
    For the case $j\ne i$ (where $j\in\{1,\ldots,p\}$), we have $M_t^{V,\phi,j}:= M_t^j$ is a $Q^{V,\phi,i}$-martingale,  by Girsanov's theorem. This is true as  $P(\tau_j = \tau_i)=0$ for all $j\in\{1,\ldots,p\}$ such that $j\ne i$ (hence,  $\langle M^j,M^i\rangle_s=0$ for $j\ne i$). 
\end{remark}

The large seller then considers the following pricing model, which takes into account their impact on the market. For a fixed pair $(V,\phi)\in\mathcal{S}^2\times\mathcal{H}^2\times\mathcal{H}_{\lambda^1}^2\times\cdots\times\mathcal{H}_{\lambda^p}^2$, which we call the wealth/risky-asset strategy pair, we have the following dynamics of the $p+1$ risky assets under the probability $Q^{V,\phi,i}$,
\begin{align*}
    dS_t^0 &= S_t^0[\mu_t^0dt + \sigma_t^0dW_t],\\
    dS_t^j &= S_{t-}^j[\mu_t^jdt+\sigma_t^jdW_t+\beta_t^jdM_t^j],\quad j\ne i,\text{ }j\in\{1,\ldots,p\}\\
    dS_t^i &=  S_{t-}^i[\mu_t^idt+\sigma_t^idW_t + \beta_t^idM_t^{V,\phi,i}].
\end{align*}

The value process $(V_t)_{t\in[0,T]}$ of the seller's portfolio associated with an initial wealth $x\in\mathbb{R}$, a risky-asset strategy $\phi$, and  an optional cumulative withdrawal process (which  the seller will choose to be equal to the optional dividend process $D$ of the option), satisfies the following dynamics,
\begin{multline}\label{eq:3_ValueProcess1}
    -dV_t = -\left(r_tV_t +\phi_t'\sigma_t\Theta_t^0 + \sum_{j=1}^p\phi_t^j\beta_t^j\Theta_t^j\lambda_t^j\right)dt + dD_t - \phi_t'\sigma_tdW_t \\- \sum_{\substack{j=1\\j\ne i}}^p\phi_t^j\beta_t^jdM_t^j - \phi_t^i\beta_t^idM_t^{V,\phi,i},\quad V_0=x,
\end{multline}
where $\phi_t'\sigma_t = \sum_{j=0}^p\phi_t^j\sigma_t^j$ and,
\begin{equation*}
    \Theta_t^0 = \frac{\mu_t^0-r_t}{\sigma_t^0},\quad\Theta_t^j = \frac{\mu_t^j-r_t-\sigma_t^j\Theta_t^0}{\beta_t^j\lambda_t^j}\mathbbm{1}_{\{\beta_t^j\lambda_t^j\ne0\}},\text{  for }j\in\{1,\ldots,p\}.
\end{equation*}

Using the expression for $M_t^{V,\phi,i}$ from \eqref{eq:3_MQLocalMartingale1}, we obtain,
\begin{multline}\label{eq:3_ValueProcess2}
    -dV_t = -\left(r_tV_t + \phi_t'\sigma_t\Theta_t^0 + \sum_{j=1}^p\phi_t^j\beta_t^j\Theta_t^j\lambda_t^j + \gamma^i(t,V_{t-},\phi_t)\lambda_t^i\phi_t^i\beta_t^i\right)dt + dD_t \\- \phi_t'\sigma_tdW_t - \sum_{j=1}^p\phi_t^j\beta_t^jdM_t^j,\quad V_0=x.
\end{multline}

By Assumption \ref{assumption:Gammai} on $\gamma^i$, for a given risky-asset strategy $\phi$,  there exists a unique process $V^{x,\phi}$ satisfying the forward SDE \eqref{eq:3_ValueProcess2} with initial condition $V_0^{x,\phi}=x$, where $x$ is the initial wealth of the investor.

We set $Z_t \coloneq \phi_t'\sigma_t$ and,  for each $j\in\{1,\ldots,p\}$,  $K_t^j\coloneq\beta_t^j\phi_t^j$. The dynamics of \eqref{eq:3_ValueProcess2} can be rewritten as,
\begin{equation}\label{eq:3_ValueProcess3}
    -dV_t = g(t,V_t,Z_t,K_t^1,\ldots,K_t^p)dt + dD_t - Z_tdW_t - \sum_{j=1}^pK_t^jdM_t^j,\quad V_0=x,
\end{equation}
where the function $g$ is defined  by,
\begin{multline*}
    g(t,y,z,k^1,\ldots,k^p) \coloneqq -r_ty - \Theta_t^0z - \sum_{j=1}^p\Theta_t^j\lambda_t^jk^j\\ - \gamma^i\Big(t,y,\frac{z-\sum_{j=1}^p\frac{k^j\sigma_t^j}{\beta_t^j}}{\sigma_t^0},\frac{k^1}{\beta_t^1},\ldots,\frac{k^p}{\beta_t^p}\Big)\lambda_t^ik^i.
\end{multline*}

If we assume that that there exists $C>0$ such that $g$ satisfies \eqref{eq:2-LAD}, then $g$ is  a $\lambda^{(p)}$-admissible driver (Definition \ref{defintion:2_LambdaDriver}).

Hence, for an option with pay-off $\eta$ at time $T$ and intermediate optional process $D$,  we have  a particular case of the pricing system $\mathbf{X}^g(\eta,D)$ from Subsection \ref{subsec:PropertiesofPricingSystem}, where $g$ is the above non-linear driver. 
From Subsection \ref{subsec:non-linearPricingCompleteMarket}, the seller's price process is  equal to $X$, where $X$ is the first component of the solution $(X,Z,K^1,\ldots,K^p)$ to the BSDE,
\begin{equation}\label{eq:3_non-linearBSDE1}
    -dX_t = g(t,X_t,Z_t,K_t^1,\ldots,K_t^p)dt + dD_t - Z_tdW_t - \sum_{j=1}^pK_t^jdM_t^j,\quad X_T = \eta.
\end{equation}

Furthermore the seller's hedging  strategy $\phi$ is obtained by the change of variables formula, 
\begin{align}\label{eq:3_non-linearChangeVariables1}
    \begin{aligned}
        \mathbf{\Phi}:\mathcal{H}_T^2\times\mathcal{H}_{\lambda^1,T}^2\times\cdots\times\mathcal{H}_{\lambda^p,T}^2&\rightarrow\mathcal{H}_T^2\times\mathcal{H}_{\lambda^1,T}^2\times\cdots\times\mathcal{H}_{\lambda^p,T}^2,\\
        (Z,K^1,\ldots,K^p)&\mapsto\mathbf{\Phi}(Z,K^1,\ldots,K^p)\coloneqq\phi=(\phi^0,\phi^1,\ldots,\phi^p),
    \end{aligned}
\end{align}
where 
\begin{equation*}
    \phi_t^0=\frac{Z_t^0 -\sum_{j=1}^p\frac{K_t^j\sigma_t^j}{\beta_t^j}}{\sigma_t^0};\quad\phi_t^j=\frac{K_t^j}{\beta_t^j}\text{ for all }j\in\{1,\ldots,p\}.
\end{equation*}

\subsection{\texorpdfstring{Example: Large Seller who Affects all $p$ Default Probabilities}{Example: Large Seller who Affects all p Default Probabilities}}\label{subsection_large2}

In this example,  we  assume that the large seller affects all $p$ default probabilities. We also assume in this example that, for each $i\in\{1,...,p\}$,  the $i$-th default intensity is bounded.

Let $\mathcal{Q}^{V,\phi}$ be the probability measure, defined by the Radon-Nikodym density process (with respect to $P$):
\begin{equation}\label{eq:3_RadonNikodym2}
    \left.\frac{d\mathcal{Q}^{V,\phi}}{dP}\right|_{\mathcal{G}_t}=\mathscr{L}_t^{V,\phi},
\end{equation}
where $\mathscr{L}^{V,\phi}$ is the solution to the SDE:
\begin{equation}\label{eq:3_DensityProcess2}
    d\mathscr{L}_t^{V,\phi} = \mathscr{L}_{t-}^{V,\phi}\left(\sum_{i=1}^p\gamma^i(t,V_{t-},\phi_t)dM_t^i\right);\quad \mathscr{L}_0^{V,\phi}=1.
\end{equation}

\begin{assumption}\label{assumption:Gammai2}
	For each $i\in\{1,\ldots,p\}$, the function $\gamma^i$ satisfies Assumption \ref{assumption:Gammai}.
\end{assumption}
By Assumption \ref{assumption:Gammai2}, Remark \ref{remark:2.12} and Proposition \ref{prop:ESM-M},  the process $\mathscr{L}^{V,\phi}$ is positive and belongs to the space $\mathcal{S}^2$.

By  Girsanov's theorem (and our assumptions),   the process $(\mathscr{W}_t^{V,\phi})$  hereafter is a $\mathcal{Q}^{V,\phi}$-Brownian motion and, for each  $i\in\{1,\ldots,p\}$, the process $(\mathscr{M}_t^{V,\phi,i})$ is a $\mathcal Q^{V,\phi}$- martingale, where 
\begin{align}
    &\begin{aligned}
        \mathscr{W}_t^{V,\phi}&\coloneqq W_t - \int_0^t\frac{d\langle W,\mathscr{L}^{V,\phi}\rangle_s}{\mathscr{L}_{s-}^{V,\phi}}= W_t, \text { and,}
    \end{aligned}\label{eq:3_WQMartingale2}
\end{align}
\begin{align}
    &\begin{aligned}
        \mathscr{M}_t^{V,\phi,i}&\coloneqq  M_t^i - \int_0^t\sum_{j=1}^p\gamma^j(s,V_{s-},\phi_s)d\langle M^j,M^i\rangle_s= M_t^i-\int_0^t\gamma^i(s,V_{s-},\phi_s)\lambda_s^ids.    \end{aligned}\label{eq:3_MQLocalMartingale2}
\end{align}
For the final equality in \eqref{eq:3_MQLocalMartingale2}, we have used that $P(\tau_i=\tau_j)=0$ for all $i,j\in\{1,\ldots,p\}$,  such that $i\ne j$. 

Hence, under the measure $\mathcal{Q}^{V,\phi}$, for each $i\in\{1,\ldots,p\}$; the $i$-th $\mathbb{G}$-default intensity process is equal to $\lambda_t^i(1+\gamma^i(t,V_{t-},\phi_t))$ since we can rewrite \eqref{eq:3_MQLocalMartingale2} as,
\begin{equation*}
    \mathscr{M}_t^{V,\phi,i}\coloneqq M_t^i - \int_0^t\gamma^i(s,V_{s-},\phi_s)\lambda_s^ids = N_t^i - \int_0^t\lambda_s^i(1+\gamma^i(s,V_{s-},\phi_s))ds.
\end{equation*}

For a given  wealth/risky-asset strategy pair $(V,\phi)\in\mathcal{S}^2\times\mathcal{H}^2\times\mathcal{H}_{\lambda^1}^2\times\cdots\times\mathcal{H}_{\lambda^p}^2$,  the $p+1$ risky assets have the   following dynamics  under $\mathcal{Q}^{V,\phi}$:
\begin{align*}
    dS_t^0 &= S_t^0[\mu_t^0dt + \sigma_t^0dW_t],\\
    dS_t^i &= S_{t-}^i[\mu_t^idt + \sigma_t^idW_t + \beta_t^id\mathscr{M}_t^{V,\phi,i}],\quad i\in\{1,\ldots,p\}.
\end{align*}

The value process $(V_t)_{t\in[0,T]}$ of the seller's portfolio associated with an initial wealth $x\in\mathbb{R}$, a risky-asset strategy $\phi$, and a cumulative withdrawal optional process  (which  the seller chooses in such a way as to be equal to the optional dividend process $D$ of the option,) satisfies the following,
\begin{multline}\label{eq:3_ValueProcess4}
    -dV_t = -\left(r_tV_t + \phi_t'\sigma_t\Theta_t^0 + \sum_{i=1}^p\phi_t^i\beta_t^i\Theta_t^i\lambda_t^i\right)dt + dD_t \\- \phi_t'\sigma_tdW_t - \sum_{i=1}^p\phi_t^i\beta_t^id\mathscr{M}_t^{V,\phi,i},\quad V_0=x,
\end{multline}
where we have $\phi_t'\sigma_t \coloneqq\sum_{i=0}^p\phi_t^i\sigma_t^i$ and,
\begin{equation*}
    \Theta_t^0=\frac{\mu_t^0-r_t}{\sigma_t^0},\quad\Theta_t^i=\frac{\mu_t^i-r_t-\sigma_t^i\Theta_t^0}{\beta_t^i\lambda_t^i}\mathbbm{1}_{\{\beta_t^i\lambda_t^i\ne0\}},\text{ for }i\in\{1,\ldots,p\}.
\end{equation*}

We  use the expression for $\mathscr{M}_t^{V,\phi,i}$ from \eqref{eq:3_MQLocalMartingale2} to obtain, 
\begin{multline}\label{eq:3_ValueProcess5}
    -dV_t = -\left(r_tV_t + \phi_t'\sigma_t\Theta_t^0 + \sum_{i=1}^p\phi_t^i\beta_t^i\Theta_t^i\lambda_t^i + \sum_{i=1}^p\gamma^i(t,V_{t-},\phi_t)\lambda_t^i\phi_t^i\beta_t^i\right)dt + dD_t \\- \phi_t'\sigma_tdW_t - \sum_{i=1}^p\phi_t^i\beta_t^idM_t^i,\quad V_0=x.
\end{multline}

By Assumption \ref{assumption:Gammai2}, for a given strategy  $\phi$,  there exists a unique process $V^{x,\phi}$ satisfying \eqref{eq:3_ValueProcess5} with initial condition $V_0^{x,\phi}=x$, where $x$ is the initial wealth of the trader.
Using \eqref{eq:3_ValueProcess5} and setting $Z_t\coloneq\phi_t'\sigma_t$ and $K_t^i\coloneq \beta_t^i\phi_t^i$, for each $i\in\{1,\ldots,p\}$,   we get,
\begin{equation}\label{eq:3_ValueProcess6}
    -dV_t = g(t,V_t,Z_t,K_t^1,\ldots,K_t^p)dt + dD_t - Z_tdW_t - \sum_{i=1}^pK_t^idM_t^i,\quad V_0=x,
\end{equation}
where the function $g$ is defined by,
\begin{multline*}
    g(t,y,z,k^1,\ldots,k^p)\coloneqq -r_ty - \Theta_t^0z - \sum_{i=1}^p\Theta_t^i\lambda_t^ik^i\\ - \sum_{i=1}^p\gamma^i\Big(t,y,\frac{z-\sum_{j=1}^p\frac{k^j\sigma_t^j}{\beta_t^j}}{\sigma_t^0},\frac{k^1}{\beta_t^1},\ldots,\frac{k^p}{\beta_t^p}\Big)\lambda_t^ik^i.
\end{multline*}
If there exists $C>0$ such that the function $g$ satisfies  condition \eqref{eq:2-LAD}, then we have another example of the pricing system $\mathbf{X}^g$ from Subsection \ref{subsec:PropertiesofPricingSystem}, where the non-linear driver $g$ is the one from above.  


\newpage


\begin{appendices}

\section{Some Technical Lemmas}\label{secA1}

\begin{lemma}\label{lemma:hDecomposition}
    Let $h$ be a non-decreasing optional rcll process, with $h_0=0$ and $\mathbb{E}[h_T^2]<\infty$, that is, $h$ is a non-decreasing process in $\mathcal{A}_T^2$. Then $h$ has at most $p$ inaccessible jumps and these jumps occur at $\tau_1,\ldots,\tau_p$. Moreover $h$ can be uniquely decomposed as follows $h_t = B_t + \Delta h_{\tau_1}\mathbbm{1}_{\{\tau_1\le t\}} + \cdots + \Delta h_{\tau_p}\mathbbm{1}_{\{\tau_p\le t\}} = B_t + \sum_{i=1}^p\int_0^t\psi_t^idN_t^i$, where $(B_t)_{t\in[0,T]}$ is a (predictable) process in $\mathcal{A}_{p,T}^2$ and for each $i\in\{1,\ldots,p\}$, $\psi^i\in\mathcal{H}_{\lambda^i,T}^2$.
\end{lemma}
\begin{proof}
    Since $h$ is a square-integrable non-decreasing optional rcll process, $h$ is a square-integrable (rcll) submartingale. Thus, by the Doob-Meyer decomposition applied to $h$, there exists a unique predictable process $a\in\mathcal{A}_{p,T}^2$ and a unique square-integrable martingale $m$ with $m_0=0$ such that $h_t = a_t + m_t$. Using the martingale representation property from Theorem \ref{theorem:MRP} Eq. \eqref{eq:2-MRP}, the $\mathbb{G}$-martingale $(m_t)_{t\in[0,T]}$ can be uniquely represented as $m_t = \int_0^tz_sdW_s + \sum_{i=1}^p\int_0^t\psi_s^idM_s^i$, where $z\in\mathcal{H}_T^2$ and $\psi^1\in\mathcal{H}_{\lambda^1,T}^2,\ldots,\psi^p\in\mathcal{H}_{\lambda^p,T}^2$. Using $dM_s^i = dN_s^i - \lambda_s^ids$ (from \eqref{eq:2.1}), we get,
    \begin{equation*}
        m_t = \int_0^tz_sdW_s - \sum_{i=1}^p\int_0^t\psi_s^i\lambda_s^ids + \sum_{i=1}^p\int_0^t\psi_s^idN_s^i.
    \end{equation*}
    Thus the process $h$ is uniquely written $h_t = a_t + \int_0^tz_sdW_s - \sum_{i=1}^p\int_0^t\psi_s^i\lambda_s^ids + \sum_{i=1}^p\int_0^t\psi_s^idN_s^i$. Setting $B_t \coloneqq a_t + \int_0^tz_sdW_s - \sum_{i=1}^p\int_0^t\psi_s^i\lambda_s^ids$, we get,
    \begin{equation*}
        h_t = B_t + \sum_{i=1}^p\int_0^t\psi_s^idN_s^i = B_t + \sum_{i=1}^p\psi_{\tau_i}^i\mathbbm{1}_{\{t\ge\tau_i\}}.
    \end{equation*}
    The process $B$ is predictable since it is the sum of predictable terms. Moreover, $B$ is square-integrable.
    
    The equality $h_t = B_t + \sum_{i=1}^p\psi_t^i\mathbbm{1}_{\{t\ge\tau_i\}}$, together with the predictability of $(B_t)$, the non-decreasingness of $h$ and the assumption that $0\le\tau_1<\tau_2<\cdots<\tau_p$ a.s. implies that $\Delta h_{\tau_1}=\psi_{\tau_1}^1\ge0$ a.s. on $\{\tau_1\le T\}$, $\Delta h_{\tau_2}=\psi_{\tau_2}^2\ge0$ a.s. on $\{\tau_2\le T\}$ and $\Delta h_{\tau_p}=\psi_{\tau_p}^p\ge0$ a.s. on $\{\tau_p\le T\}$, hence $(B_t)$ is non-decreasing.
\end{proof}

\begin{lemma}\label{lemma:nonDecreasingDecompositions}
	Let $D$ and $\hat{D}$ be \textbf{optional} processes in $\mathcal{A}_T^2$. Let $D',\hat{D}'$ in $\mathcal{A}_{p,T}^2$ and $\theta^i,\hat{\theta}^i\in\mathcal{H}_{\lambda^i,T}^2$ for $i\in\{1,\ldots,p\}$, be the unique processes such that,
\begin{align}\label{eq:lemma_nDD1}
	\begin{aligned}
		D_t = D_t' + \int_0^t\sum_{i=1}^p\theta_s^idN_s^i = D_t' + \sum_{i=1}^p\theta_{\tau_i}^i\mathbbm{1}_{\{\tau_i\le t\}},\quad\text{a.s.}\\
		\hat{D}_t = \hat{D}_t' + \int_0^t\sum_{i=1}^p\hat{\theta}_s^idN_s^i = \hat{D}_t' + \sum_{i=1}^p\hat{\theta}_{\tau_i}^i\mathbbm{1}_{\{\tau_i\le t\}},\quad\text{a.s.}
	\end{aligned}
\end{align}
If  $\bar{D}\coloneqq D - \hat{D}$ is non-decreasing, then $\bar{D}'\coloneqq D'-\hat{D}'$ is non-decreasing and for each $i\in\{1,\ldots,p\}$, $\theta_{\tau_i}^i\ge\hat{\theta}_{\tau_i}^i$ a.s. on $\{\tau_i\le T\}$.
\end{lemma}
\begin{proof}
	We note that \eqref{eq:lemma_nDD1} holds by Proposition \ref{prop:OptionalpJumps}. Using \eqref{eq:lemma_nDD1} we have,
\begin{align}\label{eq:lemma_nDD2}
	\begin{aligned}
	\bar{D}_t&\coloneqq \bar{D}_t' + \left(\sum_{i=1}^p\left(\theta_{\tau_i}^i-\hat{\theta}_{\tau_i}^i\right)\mathbbm{1}_{\{\tau_i\le t\}}\right)\\ 
	&=\left(D_t' - \hat{D}_t'\right) + \left(\sum_{i=1}^p\left(\theta_{\tau_i}^i-\hat{\theta}_{\tau_i}^i\right)\mathbbm{1}_{\{\tau_i\le t\}}\right).
	\end{aligned}
\end{align}
As an rcll predictable process does  not jump at totally inaccessible stopping times (cf. \cite{jacod2013limit}, Proposition 2.24), we have  $\Delta \bar{D}'_{\tau_i}=0$ for all $i\in\{1,\ldots,p\}$ a.s. Since, $\bar{D}$ is non-decreasing,  we have, for $i\in\{1,\ldots,p\}$,  $\Delta\bar{D}_{\tau_i} = \theta_{\tau_i}^i-\hat{\theta}_{\tau_i}^i\ge0$ a.s.\\
Let us consider each of the sets $A_k$ from the partition from \eqref{sets}.  On $A_0=\{\tau_1>T\}$,  
 we have $\bar{D}_t(\omega) = D_t'(\omega)-\hat{D}_t'(\omega)=\bar{D}'_t(\omega);$ hence, $t\mapsto \bar{D}'_t(\omega)$ is non-decreasing on  $A_0$ (as $t\mapsto \bar{D}_t(\omega)$ is non-decreasing). 
Let $k\in\{1,\ldots,p-1\}$. On $A_k$, by reasoning successively for $t\in[0,\tau_1(\omega))$ , ..., for $t\in[\tau_{k-1},\tau_k(\omega))$, and for $t\in[\tau_k(\omega),T]$,  by using Eq. \eqref{eq:lemma_nDD2} and the assumption that $\bar D$ is non-decreasing,  and by the fact that $\bar D'$ is predictable (and hence, does not jump at any of the $\tau_k$'s), we get that  $t\mapsto \bar{D}_t'(\omega)$ is non-decreasing on $A_k$. Since, $A_k$ form a partition, we conclude.   
\end{proof}

\begin{lemma}\label{lemma:expSemiMartingale}
The solution of     the following forward SDE,
    \begin{equation}\label{eq:AppendixESM1}
        d\zeta_t = \zeta_{t-}(\beta_tdW_t + \sum_{i=1}^p\gamma_t^idM_t^i);\quad\zeta_0=0,
    \end{equation}
    where the processes $M^i$ are given by \eqref{eq:2-MRP}, is: for $t\in[0,T]$,
    \begin{equation}\label{eq:AppendixESM2}
        \zeta_t = \exp\left(\int_0^t\beta_sdW_s-\frac{1}{2}\int_0^t\beta_s^2ds\right)\exp\left(-\int_0^t\sum_{i=1}^p\gamma_s^i\lambda_s^ids\right)\prod_{i=1}^p(1+\gamma_{\tau_i}^i\mathbbm{1}_{\{\tau_i\leq t\}}),\text{ a.s.}
    \end{equation}
\end{lemma}
\begin{proof}
    The SDE from  \eqref{eq:AppendixESM1} can be solved by applying the Dol\'eans-Dade formula (cf., for instance, \cite{jacod2013limit}) to the semimartingale $(X_t)_{t\in[0,T]},$ where $X_t\coloneqq\int_0^t\beta_sdW_s + \int_0^t\sum_{i=1}^p\gamma_s^idM_s^i$, (with $X_0=0$).
     We have 
    \begin{equation}\label{eq:AppendixESM3}
        \zeta_t=\mathcal{E}(X)_t = \exp\left(X_t-X_0-\frac{1}{2}[X^c]_t\right)\prod_{0<s\le t}(1+\Delta X_s)e^{-\Delta X_s},
    \end{equation}
    where
    \begin{equation}\label{eq:AppendixESM4}
        [X^c]_t = \int_0^t\beta_s^2ds.
    \end{equation}
    Since $P(\tau_i = \tau_j)=0$ for $i,j\in\{1,\ldots,p\}$ such that $i\ne j$, we get,
    \begin{align}\label{eq:AppendixESM6}
        \begin{aligned}
            \prod_{0<s\le t}(1+\Delta X_s)e^{-\Delta X_s}  
            =\exp\left(-\int_0^t\sum_{i=1}^p\gamma_s^idM_s^i - \int_0^t\sum_{i=1}^p\gamma_s^i\lambda_s^ids\right)\prod_{i=1}^p\left(1+\gamma_{\tau_i}^i\mathbbm{1}_{\{\tau_i\le t\}}\right).
        \end{aligned}
    \end{align}
    Substituting \eqref{eq:AppendixESM4} and \eqref{eq:AppendixESM6} into \eqref{eq:AppendixESM3}, we get the desired result  \eqref{eq:AppendixESM2}.
\end{proof}

%
	
\end{appendices}
\newpage


\end{document}